\documentclass[sn-mathphys-num]{sn-jnl}%
\usepackage{siunitx}
\usepackage{graphicx}%
\usepackage{multirow}%
\usepackage{amsmath,amssymb,amsfonts}%
\usepackage{mathrsfs}%
\usepackage[title]{appendix}%
\usepackage{xcolor}%
\usepackage{textcomp}%
\usepackage{manyfoot}%
\usepackage{algorithm}%
\usepackage{algorithmicx}%
\usepackage{algpseudocode}%
\usepackage{listings}%
\usepackage{natbib}
\usepackage{hyperref}
\usepackage{mathtools}
\usepackage{url}
\usepackage{amsthm}
\usepackage{soul}
\usepackage{enumitem}

\newtheorem{theorem}{Theorem}%

\newtheorem{lemma}{Lemma}

\begin{document}

\title[Article Title]{Enumeration of Tree-like Multigraphs with a Given Number of Vertices,
 Self-loops and Multiple Edges}
 \author*[1]{\fnm{Naveed Ahmed Azam}}\email{azam@amp.i.kyoto-u.ac.jp}

\author[1]{\fnm{Seemab Hayat}\email{shayat@math.qau.edu.pk}}

\affil[1]{\orgdiv{Discrete Mathematics and Computational Intelligence Laboratory,  Department of Mathematics}, \orgname{Quaid-i-Azam University}, \city{Islamabad}, \country{Pakistan}}

\abstract{ 
Counting non-isomorphic tree-like multigraphs that include self-loops and multiple edges is an important problem in combinatorial enumeration, with applications in chemical graph theory, polymer science, and network modeling. Traditional counting techniques, such as Pólya’s theorem and branching algorithms, often face limitations due to symmetry handling and computational complexity. This study presents a unified dynamic programming framework for enumerating tree-like graphs characterized by a fixed number of vertices, self-loops, and multiple edges. The proposed method utilizes canonical rooted representations and recursive decomposition of subgraphs to eliminate redundant configurations, ensuring exact counting without the need for explicit structure generation. The framework also provides analytical bounds and recurrence relations that describe the growth behaviour of such multigraphs. This work extends previous models that treated self-loops and multiple edges separately, offering a general theoretical foundation for the enumeration of complex tree-like multigraphs in both mathematical and chemical domains.}

\maketitle
\section{Introduction}
Graph enumeration with prescribed structural constraints is a fundamental problem in combinatorial mathematics with significant applications in computational chemistry, bioinformatics, and drug discovery~\cite{Blum2009, Lim2020, Meringer2013}. The problem of counting and generating chemical compounds can be naturally formulated as enumerating graphs subject to specific topological requirements~\cite{Benecke1995, Peironcely2012}. Chemical graphs often exhibit tree-like structures, either as standalone acyclic molecules or as tree-like components attached to cyclic cores, where these acyclic portions critically influence molecular properties such as solubility, reactivity, and biological activity~\cite{Batool2024, Cao2024}.  From a chemical perspective, tree-like structures with self-loops and multiple edges arise naturally in polymer topology analysis~\cite{Tezuka2002, Galina1988}. For a chemical graph, the polymer topology is obtained by repeatedly removing vertices of degree one and two, resulting in a connected graph in which each vertex has degree at least three, possibly containing self-loops and multiple edges~\cite{Tezuka2002}. The cycle rank of such a graph is defined as the number of edges that must be removed to obtain a simple spanning tree. Understanding and enumerating these structures is essential for classifying polymer topologies and elucidating structural relationships between different macromolecules and their synthetic pathways~\cite{Galina1988}.

Trees with self-loops and multiple edges represent an important class of polymer topologies. Specifically, a tree with a given number of vertices, self-loops, and multiple edges has a cycle rank and includes all polymer topologies with this tree-like structure. While previous research has addressed tree-like graphs with either self-loops~\cite{Azam2020} or multiple edges~\cite{Azam2024} separately, the simultaneous enumeration of structures containing both features remains an open and practically important problem. This combined case presents unique theoretical challenges due to the increased complexity of canonical representation and isomorphism detection when both types of cyclic structures coexist.

Several enumeration paradigms have been developed to address graph counting problems. Pólya's enumeration theorem~\cite{Polya1937, Polya2012} provides a classical combinatorial approach based on cycle indices of symmetry groups. The key idea is to construct a generating function using the cycle index of the symmetry group of the underlying object, which guarantees that isomorphic graphs are counted only once. However, computing these cycle indices becomes computationally demanding for complex graph structures, particularly when dealing with graphs containing both self-loops and multiple edges.  Branch-and-bound algorithms have been extensively studied for tree-like chemical graph enumeration. MacKay and Piperno ~\cite{McKay2014} developed Nauty and Traces, employing branch-and-bound strategies to efficiently test graph isomorphism and enumerate graphs. Akutsu and Fukagawa~\cite{Akutsu2007} developed a branch-and-bound method for enumerating tree-like chemical graphs from given frequency vectors. Fujiwara et al.~\cite{Fujiwara2008} improved this approach using the two-tree enumeration method by Nakano and Uno~\cite{Nakano2003, Nakano2005}. Ishida et al.~\cite{Ishida2008} presented branch-and-bound algorithms based on path frequencies, while Shimizu et al.~\cite{Shimizu2011} and Suzuki et al.~\cite{Suzuki2012} proposed feature-vector-based enumeration techniques for acyclic chemical graphs. Despite these advances, branch-and-bound methods share common drawbacks: they often compute the total number of solutions only after full enumeration, generate unnecessary intermediate structures, and suffer from high computational costs, making them inefficient for large or complex chemical graphs.

Dynamic programming (DP) has emerged as a more efficient paradigm for graph enumeration, addressing many inefficiencies of branch-and-bound approaches. Akutsu and Fukagawa~\cite{Akutsu2005} introduced a DP approach for inferring graphs from path frequency, demonstrating polynomial-time solutions for bounded-degree trees. Rue et al.~\cite{Rue2014} extended DP to surface-embedded graphs using surface cut decomposition. Imada et al.~\cite{Imada2011, Imada2010} developed DP approaches for enumerating stereoisomers of tree-structured molecular graphs and outerplanar graphs. Masui et al.~\cite{Masui2009} designed efficient DP algorithms for ranking, unranking, and enumerating unlabeled rooted and unrooted trees. He et al.~\cite{He2016} applied DP combined with backtracking to enumerate naphthalene-containing chemical structures. DP-based enumeration using graph-theoretic feature vectors has been explored~\cite{Azam2021}, and later extended to arbitrary chemical graphs~\cite{Azam2020b, Akutsu2020}. Azam et al.~\cite{Azam2020} proposed a DP algorithm for counting all tree-like graphs with vertices and self-loops (but no multiple edges). The method exploits Jordan's result~\cite{Jordan1869} that every tree has either a unique unicentroid or bicentroid, allowing each tree to be uniquely viewed as a rooted tree. Subsequently, they extended this work to generate such structures~\cite{Azam2020c}. More recently, Ilyas et al.~\cite{Azam2024} work on counting tree-like multigraphs with a given number of vertices and multiple edges (but no self-loops).

Despite these advances, no existing algorithm addresses the simultaneous enumeration of tree-like graphs containing both self-loops and multiple edges. This gap is particularly significant given that polymer topologies in general contain both structural features. For a connected graph, which may contain self-loops and multiple edges, the cycle rank is defined as the number of edges that must be removed from the graph to obtain a simple spanning tree. Haruna et al.~\cite{Haruna2017} proposed a method to enumerate all polymer topologies with cycle rank up to five, but did not provide tight bounds for higher cycle ranks. Establishing such bounds requires counting tree-like structures with arbitrary combinations of self-loops and multiple edges.

This work addresses the problem of counting all mutually non-isomorphic tree-like multigraphs with a prescribed number n of vertices, $s$ self-loops, and $m$ multiple edges, thereby filling the gap left by previous research that treated these features separately.

The structure of this paper is as follows. Section~\ref{preliminaries} introduces the preliminary concepts. Section~\ref{counting} presents a recursive relation for enumerating tree-like multigraphs with a given number of vertices and multiple edges, which serves as the foundation for developing a DP-based counting algorithm. Section~\ref{conclusion} provides the conclusion and discusses possible future research directions.
\section{Preliminaries}\label{preliminaries}
The term \textit{graph} refers to an undirected graph that may contain both self-loops and multiple edges unless stated otherwise. Let $G$ be a graph. We denote by $V(G)$ and $E(G)$ the vertex set and edge set of $G$, respectively. An edge connecting two vertices $u$ and $v$ in $G$ is denoted by $uv$ (or equivalently $vu$). An edge of the form $uu$ is called a self-loop on vertex $u$. For a vertex $v \in V(G)$, we denote by $s(v)$ the number of self-loops on $v$, and we define $s(G) = \sum_{v \in V(G)} s(v)$ to be the total number of self-loops in $G$. The multiplicity of an edge $uv$, denoted by $e(uv)$, is the number of additional edges beyond the first connecting vertices $u$ and $v$.  We define $e(G) = \sum_{uv \in E(G)} e(uv)$ to be the total number of multiple edges in $G$. A graph that contains neither self-loops nor multiple edges is called a simple graph, while a graph that may contain self-loops or multiple edges (or both) is called a multigraph. In this work, we focus on multigraphs that allow both self-loops and multiple edges. 

For a vertex $v \in V(G)$, the set of vertices adjacent to $v$ (excluding $v$ itself) is denoted by $N_G(v)$ and is called the neighborhood of $v$. The degree of a vertex $v$ in $G$, denoted by $\deg_G(v)$, is defined as $|N_G(v)| + 2s(v)$, where each self-loop contributes two to the degree. When the graph $G$ is clear from context, we may write $N(v)$ and $\deg(v)$ instead of $N_G(v)$ and $\deg_G(v)$. A graph $H$ is a subgraph of $G$ if $V(H) \subseteq V(G)$ and $E(H) \subseteq E(G)$. For a subset $X \subseteq V(G)$, the induced subgraph $G[X]$ is the subgraph of $G$ with vertex set $X$ and edge set consisting of all edges in $E(G)$ that have both endpoints in $X$, along with all self-loops on vertices in $X$. A path $P(u, v)$ from vertex $u$ to vertex $v$ is a sequence of vertices $w_1, w_2, \ldots, w_k$ where $w_1 = u$, $w_k = v$, and $w_iw_{i+1} \in E(G)$ for all $i = 1, \ldots, k - 1$. A path is called simple if no vertex appears more than once in the sequence. A graph $G$ is connected if there exists a path between every pair of distinct vertices in $V(G)$. A connected component of $G$ is a maximal connected subgraph of $G$. A cycle is a path that starts and ends at the same vertex, contains at least three vertices, and has no repeated edges or vertices except for the starting and ending vertex. A tree-like multigraph is defined as a connected multigraph that contains no cycles. Note that a tree-like multigraph may contain self-loops and multiple edges while still being cycle-free in the graph-theoretic sense (i.e., having no cyclic paths through distinct vertices). 

For a connected multigraph $G$ (possibly with self-loops and multiple edges), the cycle rank is defined as the minimum number of edges (including self-loops and multiple edges) that must be removed to obtain a simple spanning tree of $G$. A rooted tree-like multigraph is a tree-like multigraph with a specifically designated vertex called the root. Let $P$ be a rooted tree-like multigraph with self-loops and multiple edges. We denote by $r_P$ the root of $P$. For a vertex $v \in V(P) \setminus \{r_P\}$, an ancestor of $v$ is any vertex (other than $v$ itself) that lies on the unique simple path $P(v, r_P)$ from $v$ to the root $r_P$. If $u$ is an ancestor of $v$, then $v$ is called a descendant of $u$. For a vertex $v \in V(P) \setminus \{r_P\}$, the parent of $v$, denoted by $p(v)$, is the unique ancestor $u$ of $v$ such that $uv \in E(P)$. We call $v$ a child of $p(v)$, and vertices with the same parent are called siblings. For a vertex $v \in V(P)$, we denote by $P_v$ the descendant subgraph of $P$ rooted at $v$, which is the subgraph induced by $v$ and all its descendants. Note that $P_v$ is a maximal rooted tree-like multigraph of $P$, rooted at $v$ and containing all structural features (self-loops and multiple edges) present in that portion of $P$. We denote by $v(P_v)$ the number of vertices, by $s(P_v)$ the number of self-loops, and by $e(P_v)$ the number of multiple edges in the descendant subgraph $P_v$. Additionally, for any child $v$ of the root $r_P$, we denote by $e(r_Pv)$ the multiplicity of the edge connecting $r_P$ and $v$ (i.e., the number of multiple edges between the root and $v$). 

According to Jordan~\cite{Jordan1869}, every simple tree with $n \ge 1$ vertices has either a unique vertex (called the unicentroid) or a unique edge (called the bicentroid) such that removing it divides the tree into connected components with specific size properties. The unicentroid is a vertex whose removal creates components each containing at most $\lfloor (n - 1)/2 \rfloor$ vertices. The bicentroid is an edge whose removal creates exactly two components, each containing exactly $n/2$ vertices; such an edge exists only when $n$ is even. Two rooted tree-like multigraphs $P$ and $Q$ with roots $r_P$ and $r_Q$, respectively, are said to be isomorphic if there exists a bijective mapping $\phi: V(P) \to V(Q)$ such that $\phi(r_P) = r_Q$, and for each vertex $v \in V(P)$ it holds that $s(v) = s(\phi(v))$, $uv \in E(P)$ if and only if $\phi(u)\phi(v) \in E(Q)$, and $e(uv) = e(\phi(u)\phi(v))$ for all $u, v \in V(P)$. 

Two unrooted tree-like multigraphs are isomorphic if their canonical rooted forms are isomorphic. Let $n \ge 1$, $s \ge 0$, and $m \ge 0$ be any three integers.  We denote by $\mathcal{P}(n, s, m)$ the set of all mutually non-isomorphic rooted tree-like multigraphs with $n$ vertices, $s$ self-loops, and $m$ multiple edges, and by $p(n, s, m)$ its cardinality, that is, $p(n, s, m) = |\mathcal{P}(n, s, m)|$.

\section{Proposed Method}\label{counting}
We aim to determine the exact count $p(n, s, m)$ of non-isomorphic rooted tree-like multigraphs characterized by a fixed number of vertices and multiple edges, excluding self-loops.  We define the boundary cases as $p(1, 0, 0) = 1$, corresponding to a single vertex with no self-loops or multiple edges, $p(1, s, m) = 0$ for any $s \ge 1$ or $m \ge 1$, since a single vertex cannot have multiple edges, and $p(n, s, m) = 0$ for $n = 0$, as the empty graph is not considered in this case.
\subsection{Canonical representation}
For an enumeration method, it is essential to avoid generating isomorphic graphs. To achieve this, we define a canonical representation for graphs in $\mathcal{P}(n, s, m)$. This representation incorporates information about the number of vertices, self-loops, and multiple edges in each descendant subgraph, as well as the number of multiple edges and self-loops incident to the root in the corresponding underlying multigraphs. Formally, the canonical (or ordered) representation of a multigraph $P \in \mathcal{P}(n, s, m)$ is defined as an ordered rooted multigraph in which the descendant subgraphs of each vertex are arranged from left to right according to the following criteria. Consider a vertex $u$ of $P$ and the descendant subgraphs corresponding to its children. 

(a) First, order the descendant subgraphs in non-increasing order of the number of vertices $v(P_v)$.  
(b) If two or more descendant subgraphs have the same number of vertices, they are ordered in non-increasing order of the number of self-loops $s(P_v)$.  
(c) If two or more descendant subgraphs have both the same number of vertices and the same number of self-loops, they are ordered in non-increasing order of the number of multiple edges $e(P_v)$.  
(d) Finally, if two or more subgraphs remain equivalent under the above criteria, they are ordered in non-increasing order of the multiplicity $e(u v)$ of the edges connecting the parent $u$ to its child $v$, and, if still equal, in non-increasing order of the number of self-loops on the child vertex $s(v)$.

The canonical representation of a tree-like multigraph $P \in \mathcal{P}(n, s, m)$ is thus the unique ordered representation that satisfies conditions (a)–(d) recursively for every vertex. This canonical form ensures that isomorphic tree-like multigraphs are represented identically, thereby enabling enumeration without duplication.

\subsection{Subproblem} 
We define the subproblems for DP based on the descendant subgraphs with the maximum number of vertices, self-loops, and multiple edges.  
For this purpose we define the following terms for each $P \in \mathcal{P}(n, \Delta)$, 
\begin{align*}
{\rm Max}_{\rm v}(P) & \triangleq \max\{\{ {\rm v}(P_{v}): v\in N(r_{P})\}\cup \{0\}\},\\
{\rm Max}_{\rm s}(P) & \triangleq  \max\{\{ {\rm s}(P_{v}): v\in N(r_{P}),  {\rm v}(P_{v}) ={\rm Max}_{\rm v}(P)\} \cup\{0\}\}\\
{\rm Max}_{\rm m}(P) & \triangleq \max \{\{{\rm e}(P_{v}): v\in N(r_{P}), {\rm v}(P_{v}) ={\rm Max}_{\rm v}(P)\} \cup\{0\}\},\\
{\rm Max}_{\rm l}(P) & \triangleq \max \{\{{\rm e}(r_{p} r_{P_{v}}): {\rm v}(P_{v})={\rm Max}_{\rm v}(P)\ \text{ and } {\ {\rm e}(P_{v})} ={\rm Max}_{\rm p}(P)\} \cup \{0\}\}.
\end{align*}
Intuitively, ${\rm Max}_{\rm v}(P)$ denotes the maximum number of vertices in the descendant subgraphs. ${\rm Max}_{\rm s}(P)$ denotes the maximum number of self-loops incident to a vertex $v$ in the descendant subgraphs.
${\rm Max}_{\rm m}(P)$ denotes the maximum number of multiple edges in the descendant subgraphs containing ${\rm Max}_{\rm v}(P)$ vertices. Finally, ${\rm Max}_{\rm l}(M)$ denote the maximum number of multiple edges between the root of $M$ and the roots of the descendant subgraphs containing  ${\rm Max}_{\rm v}(P)$ vertices, ${\rm Max}_{\rm s}(P)$ self-loops,  and ${\rm Max}_{\rm m}(P)$ multiple edges. 

Note that for any multigraph $P \in \mathcal{P}(1, s, m)$, it holds that  ${\rm Max}_{\rm v}(P) = 0$, ${\rm Max}_{\rm s}(P)=0$, $ {\rm Max}_{\rm m}(P) = 0$ and ${\rm Max}_{\rm l}(P) = 0$.

Let $f, g, h, k$ be any four integers such that $f \geq 1$ and $g, h, k \geq 0$. 
We define a subproblem of $\mathcal{P}(n, s, m)$ for DP as follows:
\begin{align*}
\mathcal{P}(n, s, m ,{f}_{\leq}, {g}_{\leq}, {h}_{\leq}, {k}_{\leq})
&\triangleq 
\bigl\{P\in \mathcal{P}(n, s, m) : 
{\rm Max}_{\rm v}(P) \leq f,\ 
{\rm Max}_{\rm s}(P)\leq g,\  \\
&\hspace{0.6cm}
{\rm Max}_{\rm m}(P)\leq h, {\rm Max}_{\rm l}(P)\leq k \bigr\}.
\end{align*}

Note that by the definition of $\mathcal{P}(n, s, m, {f}_{\leq}, {g}_{\leq}, {h}_{\leq}, {k}_{\leq} )$ it holds that 
\begin{enumerate}
[label=\textnormal{(\roman*)}, ref=(\roman*), font=\upshape]
\item $\mathcal{P}(n, s, m, {f}_{\leq}, {g}_{\leq}, {h}_{\leq}, {k}_{\leq} ) = \mathcal{P}(n, s, m, {n-1}_{\leq}, {g}_{\leq}, {h}_{\leq}, {k}_{\leq} )$ if $f \geq n$;
\item $\mathcal{P}(n, s, m, {f}_{\leq}, {g}_{\leq}, {h}_{\leq}, {k}_{\leq} ) =
\mathcal{P}(n, s, m, {f}_{\leq}, {s}_{\leq}, {h}_{\leq}, {k}_{\leq} )$ if $g \geq{ s + 1}$; 
\item $\mathcal{P}(n, s, m, {f}_{\leq}, {g}_{\leq}, {h}_{\leq}, {k}_{\leq} ) =
\mathcal{P}(n, s, m, {f}_{\leq}, {g}_{\leq}, {m}_{\leq}, {k}_{\leq} )$ if $h \geq{ m + 1}$;
\item $\mathcal{P}(n, s, m, {f}_{\leq}, {g}_{\leq}, {h}_{\leq}, {k}_{\leq} ) =
\mathcal{P}(n, s, m, {f}_{\leq}, {g}_{\leq}, {h}_{\leq}, {m}_{\leq} )$ if $k \geq{ m + 1}$;  
\item $\mathcal{P}(n, s, m)= \mathcal{P}(n, s, m, {n-1}_{\leq}, {s}_{\leq}, {m}_{\leq}, {m}_{\leq} )$.
\end{enumerate}
Consequently, we assume that $f \leq n - 1$, $g \leq s$, and $ h + \ k \leq m$.
Moreover, by the definition, it holds that 
$\mathcal{P}(n, s, m, {f}_{\leq}, {g}_{\leq}, {h}_{\leq}, {k}_{\leq} ) \neq \emptyset$ if $``n=1$, $s=0$, and $m=0''$  or $``n-1 \geq f \geq 1"$ and $\mathcal{M}(n, \Delta, {k}_\leq, {d}_\leq, {\ell}_\leq) = \emptyset$ 
if $``n=1$, $s \geq 1$ and $m \geq 1"$ or $``n \geq2$ and $f = 0"$.

We define the subproblems of $\mathcal{P}(n, s, m, {f}_{\leq}, {g}_{\leq}, {h}_{\leq}, {k}_{\leq} )$ based on the maximum number of vertices in the descendant subgraphs:
\begin{equation*}
\mathcal{P}(n, s, m, {f}_{=}, {g}_{\leq}, {h}_{\leq}, {k}_{\leq} )\triangleq 
\{P\in \mathcal{P}(n, s, m, {f}_{\leq}, {g}_{\leq}, {h}_{\leq}, {k}_{\leq} ) :
{\rm Max}_{\rm v}(P) = f\}.
\end{equation*}%
\noindent 
It follows from the definition of $\mathcal{P}(n, s, m, {f}_{=}, {g}_{\leq}, {h}_{\leq}, {k}_{\leq} )$ that $\mathcal{P}(n, s, m, {f}_{=}, {g}_{\leq}, {h}_{\leq}, {k}_{\leq} ) \neq \emptyset$ if $``n = 1$, $s=0$, and $m = 0"$ or $``n-1 \geq f \geq 1"$, and $\mathcal{P}(n, s, m, {f}_{=}, {g}_{\leq}, {h}_{\leq}, {k}_{\leq} ) = \emptyset$ if $``n=1$, $s \geq 1$, and $m \geq 1"$ or $`` n \geq2$ and $f = 0."$ In addition, we have the following recursive relation:
\begin{align}
&\mathcal{P}(n, s, m, {f}_{\leq}, {g}_{\leq}, {h}_{\leq}, {k}_{\leq} )= 
\mathcal{P}(n, s, m, {0}_{=}, {g}_{\leq}, {h}_{\leq}, {k}_{\leq} ) \text{ if } f = 0, \label{equ1} \\ 
&\mathcal{P}(n, s, m, {f}_{\leq}, {g}_{\leq}, {h}_{\leq}, {k}_{\leq} )= 
\mathcal{P}(n, s, m, {f-1}_{\leq}, {g}_{\leq}, {h}_{\leq}, {k}_{\leq} )\cup \mathcal{P}(n, s, m, {f}_{=}, {g}_{\leq}, {h}_{\leq}, {k}_{\leq} )  \nonumber\\
&\hspace{4.3cm}
\text { if } f \geq 1, \label{equ1p}
\end{align}
where
$\mathcal{P}(n, s, m, {f-1}_{\leq}, {g}_{\leq}, {h}_{\leq}, {k}_{\leq} )\cap 
\mathcal{P}(n, s, m, {f}_{=}, {g}_{\leq}, {h}_{\leq}, {k}_{\leq} ) = \emptyset$ for $f \geq 1$.

We define the subproblems of $\mathcal{P}(n, s, m, {f}_{=}, {g}_{\leq}, {h}_{\leq}, {k}_{\leq} )$ based on the maximum number of self-loops in the descendant subgraphs:
\begin{equation*}
\mathcal{P}(n, s, m, {f}_{=}, {g}_{=}, {h}_{\leq}, {k}_{\leq} )\triangleq 
\{P\in \mathcal{P}(n, s, m, {f}_{=}, {g}_{\leq}, {h}_{\leq}, {k}_{\leq} ) :
{\rm Max}_{\rm s}(P) = g\}.
\end{equation*}%
\noindent 
It follows from the definition of $\mathcal{P}(n, s, m, {f}_{=}, {g}_{=}, {h}_{\leq}, {k}_{\leq} )$ that $\mathcal{P}(n, s, m, {f}_{=}, {g}_{=}, {h}_{\leq}, {k}_{\leq} ) \neq \emptyset$ if $``n = 1$, $s=0$, and $m = 0"$ or $``n-1 \geq f \geq 1"$, and $\mathcal{P}(n, s, m, {f}_{=}, {g}_{=}, {h}_{\leq}, {k}_{\leq} ) = \emptyset$ if $``n=1$, $s \geq 1$, and $m \geq 1"$ or $`` n \geq2$ and $f = 0."$ In addition, we have the following recursive relation:
\begin{align}
&\mathcal{P}(n, s, m, {f}_{=}, {g}_{\leq}, {h}_{\leq}, {k}_{\leq} )= 
\mathcal{P}(n, s, m, {f}_{=}, {0}_{=}, {h}_{\leq}, {k}_{\leq} ) \text{ if } g = 0, \label{equ2} \\ 
&\mathcal{P}(n, s, m, {f}_{=}, {g}_{\leq}, {h}_{\leq}, {k}_{\leq} )= 
\mathcal{P}(n, s, m, {f}_{=}, {g-1}_{\leq}, {h}_{\leq}, {k}_{\leq} )\cup \mathcal{P}(n, s, m, {f}_{=}, {g}_{=}, {h}_{\leq}, {k}_{\leq} )\nonumber\\
&\hspace{4.3cm} 
\text { if } g \geq 1, \label{equ2p}
\end{align}
where
$\mathcal{P}(n, s, m, {f}_{=}, {g-1}_{\leq}, {h}_{\leq}, {k}_{\leq} )\cap 
\mathcal{P}(n, s, m, {f}_{=}, {g}_{=}, {h}_{\leq}, {k}_{\leq} ) = \emptyset$ for $g \geq 1$.


We define the subproblems of $\mathcal{P}(n, s, m, {f}_{=}, {g}_{=}, {h}_{=}, {k}_{\leq} )$ based on the maximum number of multiple edges in the descendant subgraphs:
\begin{equation*}
\mathcal{P}(n, s, m, {f}_{=}, {g}_{=}, {h}_{=}, {k}_{\leq} )\triangleq 
\{P\in \mathcal{P}(n, s, m, {f}_{=}, {g}_{=}, {h}_{\leq}, {k}_{\leq} ) :
{\rm Max}_{\rm m}(P) = h\}.
\end{equation*}%
\noindent 
It follows from the definition of $\mathcal{P}(n, s, m, {f}_{=}, {g}_{=}, {h}_{=}, {k}_{\leq} )$ that $\mathcal{P}(n, s, m, {f}_{=}, {g}_{=}, {h}_{=}, {k}_{\leq} ) \neq \emptyset$ if $``n = 1$, $s=0$, and $m = 0"$ or $``n-1 \geq f \geq 1"$, and $\mathcal{P}(n, s, m, {f}_{=}, {g}_{=}, {h}_{=}, {k}_{\leq} ) = \emptyset$ if $``n=1$, $s \geq 1$, and $m \geq 1"$ or $`` n \geq2$ and $f = 0."$ In addition, we have the following recursive relation:
\begin{align}
&\mathcal{P}(n, s, m, {f}_{=}, {g}_{=}, {h}_{\leq}, {k}_{\leq} )= 
\mathcal{P}(n, s, m, {f}_{=}, {g}_{=}, {0}_{=}, {k}_{\leq} ) \text{ if } h = 0, \label{equ3} \\ 
&\mathcal{P}(n, s, m, {f}_{=}, {g}_{=}, {h}_{\leq}, {k}_{\leq} )= 
\mathcal{P}(n, s, m, {f}_{=}, {g}_{=}, {h-1}_{\leq}, {k}_{\leq} )\cup \mathcal{P}(n, s, m, {f}_{=}, {g}_{=}, {h}_{=}, {k}_{\leq} ) \nonumber\\
&\hspace{4.3cm}
\text { if } h \geq 1, \label{equ3p}
\end{align}
where
$\mathcal{P}(n, s, m, {f}_{=}, {g}_{=}, {h-1}_{\leq}, {k}_{\leq} )\cap 
\mathcal{P}(n, s, m, {f}_{=}, {g}_{=}, {h}_{=}, {k}_{\leq} ) = \emptyset$ for $h \geq 1$.


We define the subproblems of $\mathcal{P}(n, s, m, {f}_{=}, {g}_{=}, {h}_{=}, {k}_{\leq} )$ based on the maximum number of multiple edges between roots and their children  in the  multigraphs:
\begin{equation*}
\mathcal{P}(n, s, m, {f}_{=}, {g}_{=}, {h}_{=}, {k}_{=} )\triangleq 
\{P\in \mathcal{P}(n, s, m, {f}_{=}, {g}_{=}, {h}_{=}, {k}_{\leq} ) :
{\rm Max}_{\rm l}(P) = k\}.
\end{equation*}%
\noindent 
It follows from the definition of $\mathcal{P}(n, s, m, {f}_{=}, {g}_{=}, {h}_{=}, {k}_{=} )$ that $\mathcal{P}(n, s, m, {f}_{=}, {g}_{=}, {h}_{=}, {k}_{=} ) \neq \emptyset$ if $``n = 1$, $s=0$, and $m = 0"$ or $``n-1 \geq f \geq 1"$, and $\mathcal{P}(n, s, m, {f}_{=}, {g}_{=}, {h}_{=}, {k}_{=} ) = \emptyset$ if $``n=1$, $s \geq 1$, and $m \geq 1"$ or $`` n \geq2$ and $f = 0."$ In addition, we have the following recursive relation:
\begin{align}
&\mathcal{P}(n, s, m, {f}_{=}, {g}_{=}, {h}_{=}, {k}_{\leq} )= 
\mathcal{P}(n, s, m, {f}_{=}, {g}_{=}, {h}_{=}, {0}_{=} ) \text{ if } k = 0, \label{equ4} \\ 
&\mathcal{P}(n, s, m, {f}_{=}, {g}_{=}, {h}_{=}, {k}_{\leq} )= 
\mathcal{P}(n, s, m, {f}_{=}, {g}_{=}, {h}_{\leq}, {k-1}_{\leq} )\cup \mathcal{P}(n, s, m, {f}_{=}, {g}_{=}, {h}_{=}, {k}_{=} ) \nonumber\\
&\hspace{4.3cm}
\text { if } k \geq 1, \label{equ4p}
\end{align}
where
$\mathcal{P}(n, s, m, {f}_{=}, {g}_{=}, {h}_{\leq}, {k-1}_{\leq} )\cap 
\mathcal{P}(n, s, m, {f}_{=}, {g}_{=}, {h}_{=}, {k}_{=} ) = \emptyset$ for $h \geq 1$.

\subsection{Recursive relation}
To obtain recursive relations for computing the size of 
$\mathcal{P}(n, s, m, {f}_{=}, {g}_{=}, {h}_{=}, {k}_{=} )$, it is essential to analyze the types of descendant subgraphs. 
Let $P$ be a multigraph such that $P \in \mathcal{P}(n, s, m, {f}_{=}, {g}_{=}, {h}_{=}, {k}_{=} )$. 
Then it is easy to observe that for any vertex $v\in N(r_{p})$, the descendant subgraph $P_{v}$ satisfies exactly one of the following four conditions:

\begin{enumerate}[label=(C\arabic*), ref=C\arabic*]

\item \label{con:condit11}$ {\rm v}(P_{v}) = f$, $ {\rm s}(P_{v}) = g$, ${\rm e}(P_{v}) = h$ and ${\rm e}({r_{P} r_{p_{v}}}) = k$. 

\item \label{con:condit12}$ {\rm v}(P_{v}) = f$, $ {\rm s}(P_{v}) = g$, ${\rm e}(P_{v}) = h$ and $0 \leq {\rm e}({r_{P} r_{p_{v}}}) < k$. 

\item \label{con:condit13}$ {\rm v}(P_{v}) = f$, $ {\rm s}(P_{v}) = g$, $0 \leq {\rm e}(P_{v}) < h$ and $0 \leq {\rm e}({r_{P} r_{p_{v}}}) \leq m$. 

\item \label{con:condit14}$ {\rm v}(P_{v}) = f$, $0 \leq {\rm s}(P_{v}) < g$, $0 \leq {\rm e}(P_{v}) \leq m$ and $0 \leq {\rm e}({r_{P} r_{p_{v}}}) \leq m$.

\item \label{con:condit15}$ {\rm v}(P_{v}) < f$, $0 \leq {\rm s}(P_{v}) \leq s$, $0 \leq {\rm e}(P_{v}) \leq m$ and $0 \leq {\rm e}({r_{P} r_{p_{v}}}) \leq m$.\\

\end{enumerate}
Further, an example of a recursive relation is illustrated in Fig.~\ref{fig:relation}.
\begin{figure}[h!]
    \centering
    \includegraphics[width=0.5\textwidth]{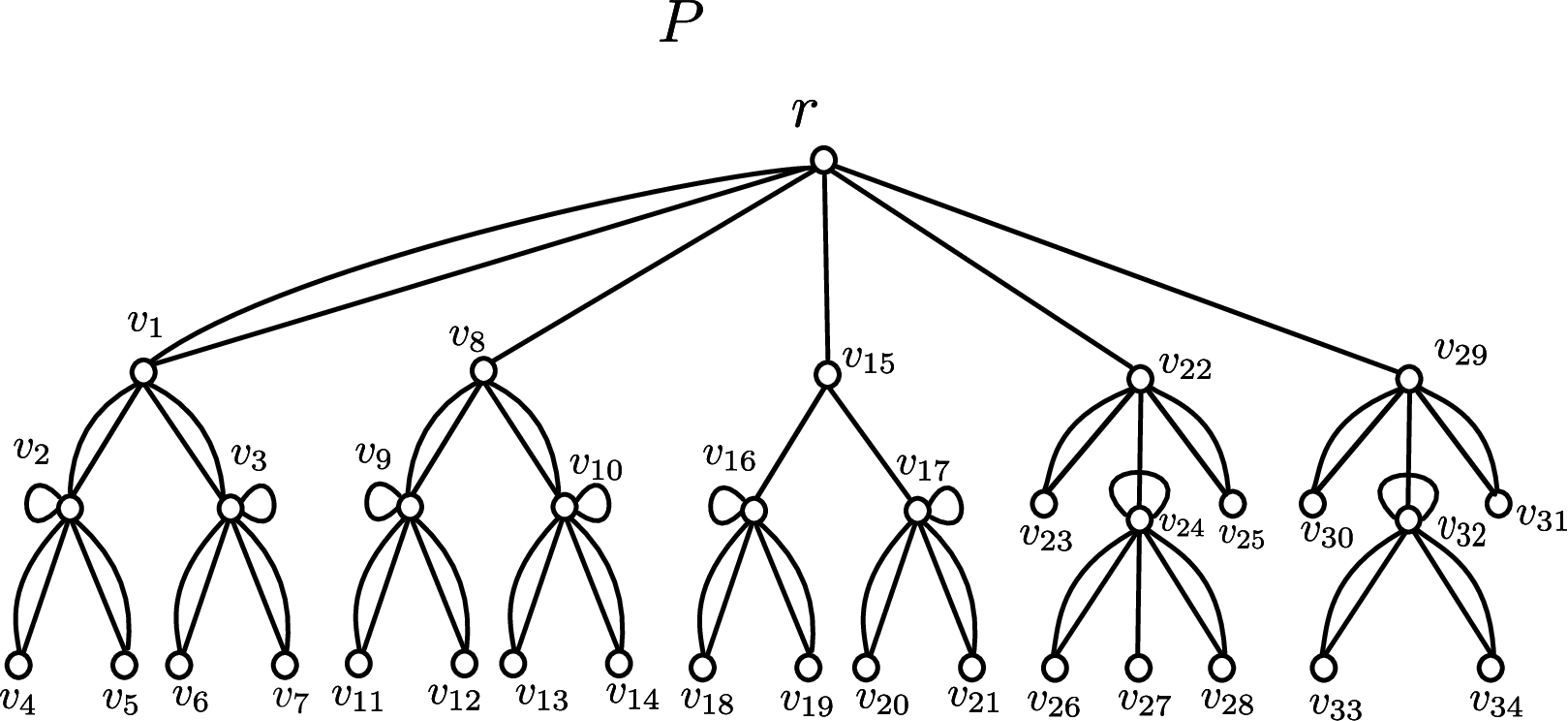}
    \caption{Example of recursive relation.}
    \label{fig:relation}
\end{figure}

\noindent We define the residual of $P \in \mathcal{P}(n, s, m, {f}_{=}, {g}_{=}, {h}_{=}, {k}_{=} )$ to be multigraph rooted at $r_P$ that is induced by the vertices 
$V(P) \setminus  \bigcup\limits_{\mathclap{\substack{v \in N(r_P),\\ 
 P_v \in \mathcal{P}(f,g, h, \,f{-}1_\le,\,g_\le,\,h_\le, h_\le)}}}{ V}(P_v).$  Basically, a residual tree-like graph consists of those descendant subgraphs whose structure is not clear. The residual tree of a tree $P$ has at least one vertex, i.e., the root of $P$.
 An illustration of a residual multigraph is given in Fig.~\ref{fig:recursive_rel}. 
\begin{figure}[h!]
    \centering
    \includegraphics[width=0.9\textwidth]{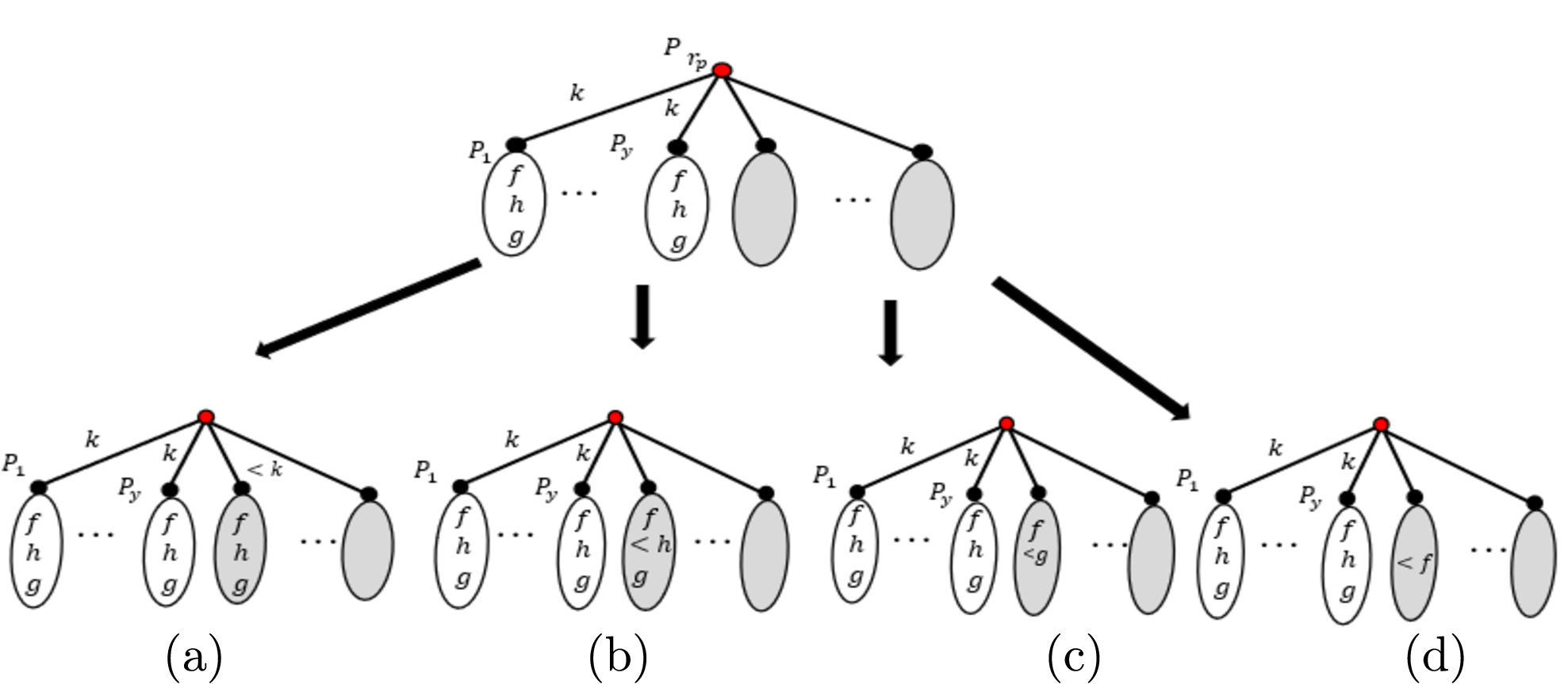}
    \caption{An illustration of the residual tree depicted with a dashed line, and the descendant subgraph is depicted with solid lines.}
    \label{fig:recursive_rel}
\end{figure}
\begin{lemma}
\label{lem_1}
For any seven integers $n \geq 3$, $f \geq 1$, $s \geq g \geq 0$, $m \geq h+k \geq 0$, and a tree-like multigraph $P\in 
\mathcal{P}(n, s, m, {f}_{=}, {g}_{=}, {h}_{=}, {k}_{=} ),$ let~$y = |\{ v\in N(r_{P}) :
 P_v \in \mathcal{P}(f,g, h, \,f{-}1_\le,\,g_\le,\,h_\le, h_\le)\}|$. 
 Then it holds that
\begin{enumerate}[label=(\roman*), ref  = (\roman*), font=\upshape]
\item $1 \leq y \leq \lfloor n-1 / f \rfloor$  with  $y \leq \lfloor s / g \rfloor$ when  $g \geq 1$ and $y \leq \lfloor m / (h+k) \rfloor$ when $h+k \geq 1 $. \label{L_1_i}
\item \label{L_1_ii} The residual tree of $P$ belongs to exactly one of the following families:\\
$\mathcal{P}(n-yf, s-yg, m-y(h+k), {f}_{=}, {g}_{=}, {h}_{=}, min \{ m-y(h+k), k-1\}_{\leq} );$\\
$\mathcal{P}(n-yf, s-yg, m-y(h+k), {f}_{=}, {g}_{=}, min \{ m-y(h+k), h-1\}_{\leq}, m-y(h+k)_{\leq} );$\\
$\mathcal{P}(n-yf, s-yg, m-y(h+k), {f}_{=}, min \{ s-yg, g-1\}_{\leq}, m-y(h+k)_{\leq} , m-y(h+k)_{\leq} );$\\
$\mathcal{P}(n-yf, s-yg, m-y(h+k), min \{ n-yf-1, f-1\}_{\leq},  s-yg_{\leq}, m-y(h+k)_{\leq} , m-y(h+k)_{\leq} ).$
\end{enumerate}
\end{lemma}
\begin{proof}
\begin{enumerate}[label=\textnormal{(\roman*)}, ref=(\roman*), font=\upshape]
\item Since $P\in 
\mathcal{P}(n, s, m, {f}_{=}, {g}_{=}, {h}_{=}, {k}_{=} )$ there exist at least $v \in N(r_p)$ such that $P_v \in \mathcal{P}(f,g, h, \,f{-}1_\le,\,g_\le,\,h_\le, h_\le)$. It follows that $y \geq 1$. It further assert that $n-1 \geq yf$, $s \geq yg$ and $m \geq y(h+k)$. This concludes that $ y \leq \lfloor n-1 / f \rfloor$  with  $y \leq \lfloor s / g \rfloor$ when  $g \geq 1$ and $y \leq \lfloor m / (h+k) \rfloor$ when $h+k \geq 1 $.
\item Let $F$ represents the residual tree of $P$. 
By the definition of $F$, it holds that $F \in \mathcal{P}(n-yf, s-yg, m-y(h+k), n-yf-1_{\leq},  s-yg_{\leq}, m-y(h+k)_{\leq} , m-y(h+k)_{\leq} )$. Moreover, for each vertex $v\in N(r_{P})\cap V(F),$ the descendant subgraph $P_{v}$ satisfies exactly one of the Conditions~\ref{con:condit12}-\ref{con:condit15}. Now, if there exists a vertex $v\in N(r_{P})\cap V(F)$ such that $P_v$ satisfies Condition~\ref{con:condit12} as illustrated in Fig.~\ref{fig:recursive_rel}(a), then $f-1 \geq 0$ and $F \in \mathcal{P}(n-yf, s-yg, m-y(h+k), {f}_{=}, {g}_{=}, {h}_{=}, min \{ m-y(h+k), k-1\}_{\leq} )$.
For $v\in N(r_{P}) \cap V(F)$ such that $P_v$ satisfies the condition ~\ref{con:condit13} as illustrated in Fig.~\ref{fig:recursive_rel}(b) i.e., $v(P_v)=f$, $s(P_v)=g$, $0 \leq e(P_v) \leq min\{m-y(h+k), h-1\}$ and $0 \leq e(r_{p}r_{P_{v}} \leq m- y(h+k)$, then residual tree $F \in \mathcal{P}(n-yf, s-yg, m-y(h+k), {f}_{=}, {g}_{=}, min \{ m-y(h+k), h-1\}_{\leq}, m-y(h+k)_{\leq} )$. For $v\in N(r_{P}) \cap V(F)$ such that $P_v$ satisfies the condition ~\ref{con:condit14} as illustrated in Fig.~\ref{fig:recursive_rel}(c) i.e., $g-1 \geq 0$ hence $v(P_v)=f$, $0 \leq s(P_v) \leq min\{m-yg, g-1\}$, $0 \leq e(P_{v}) \leq m- y(h+k)$,  and $0 \leq e(r_{p}r_{P_{v}} \leq m- y(h+k)$, then residual tree $F \in \mathcal{P}(n-yf, s-yg, m-y(h+k), {f}_{=}, \{ m-yg, g-1\}_{\leq}, m-y(h+k)_{\leq}, m-y(h+k)_{\leq} )$. If $P_v$ satisfies the ~\ref{con:condit15} as illustrated in Fig.~\ref{fig:recursive_rel}(c), $v(P_v)< f$, $0 \leq s(P_v) \leq m-yg$, $0 \leq e(P_{v}) \leq m- y(h+k)$,  and $0 \leq e(r_{p}r_{P_{v}} \leq m- y(h+k)$, then by definition of $F$, it holds that   $F \in \mathcal{P}(n-yf, s-yg, m-y(h+k), min\{n-fy-1, f-1\}_{\leq}, \{ m-yg, g-1\}_{\leq}, m-y(h+k)_{\leq}, m-y(h+k)_{\leq} )$.
\end{enumerate}
\end{proof} 

Let $n,s,m, f, g, h$ and $k$ be any seven non-negative integers such that $n-1 \geq f \geq 0$, $s \geq g \geq 0$, and $m  \geq h+k \geq 0$. Let $p(n, s, m, {f}_{\leq}, {g}_{\leq}, {h}_{\leq}, {k}_{\leq} )$, $p(n, s, m, {f}_{=}, {g}_{\leq}, {h}_{\leq}, {k}_{\leq} )$, $p(n, s, m, {f}_{=}, {g}_{=}, {h}_{\leq}, {k}_{\leq} )$, $p(n, s, m, {f}_{=}, {g}_{=}, {h}_{=}, {k}_{\leq} )$ and $p(n, s, m, {f}_{=}, {g}_{=}, {h}_{=}, {k}_{=} )$ represent the families $\mathcal{P}(n, s, m, {f}_{\leq}, {g}_{\leq}, {h}_{\leq}, {k}_{\leq} )$, $\mathcal{P}(n, s, m, {f}_{=}, {g}_{\leq}, {h}_{\leq}, {k}_{\leq} )$, $\mathcal{P}(n, s, m, {f}_{=}, {g}_{=}, {h}_{\leq}, {k}_{\leq} )$, $\mathcal{P}(n, s, m, {f}_{=}, {g}_{=}, {h}_{=}, {k}_{\leq} )$ and $\mathcal{P}(n, s, m, {f}_{=}, {g}_{=}, {h}_{=}, {k}_{=} )$. For any eight integers, $n \geq 3$, $f \geq 1$, $s \geq g \geq 0$, $m \geq h+k \geq 0$ and $z \geq 0$, let
 $$w(f, g, h ; z) \triangleq \binom{p(f,g, h, \,f{-}1_\le,\,g_\le,\,h_\le, h_\le)+ z - 1}{z}$$ denote the number of combinations with repetition of $z$ descendant subgraphs from the family 
$\mathcal{P}(f,g, h, \,f{-}1_\le,\,g_\le,\,h_\le, h_\le)$. 
With all necessary information, we are ready to discuss a recursive relation for 
$p(n, s, m, {f}_{=}, {g}_{=}, {h}_{=}, {k}_{=} )$ in~Lemma~\ref{lem_2}.
\begin{lemma}
\label{lem_2}
 For any eight integers, $n \geq 3$, $f \geq 1$, $s \geq g \geq 0$, $m \geq h+k \geq 0$ and $y \geq 0$, such that $1 \leq y \leq \lfloor n-1 / f \rfloor$  with  $y \leq \lfloor s / g \rfloor$ when  $g \geq 1$ and $y \leq \lfloor m / (h+k) \rfloor$ when $h+k \geq 1 $, it holds that
\begin{enumerate}[label=(\roman*), ref  = (\roman*), font=\upshape]
\item \label{2_i} $p(n, s, m, {f}_{=}, {g}_{=}, {h}_{=}, {k}_{=} )= \sum_{y}w(f,g,h; y)(p(n-yf, s, m, min\{n-yf-1, f-1\}_{\leq}, s_{\leq}, m_{\leq}, m_{\leq}))$ if $g=0, h=0,$ and $k=0$.

\item \label{2_ii} $p(n, s, m, {f}_{=}, {g}_{=}, {h}_{=}, {k}_{=} )= \sum_{y}w(f,g,h; y)( p(n-yf, s-yg, m-yh, f_{=}, g_{=}, min\{m-yh, h-1\}_{\leq}, m-yh_{\leq})+ p(n-yf, s-yg, m-yh, f_{=}, min\{m-yg, g-1\}_{\leq},m-yh_{\leq}, m-yh_{\leq})+ p(n-yf, s-yg, m-yh, min\{n-yf-1, f-1\}_{\leq}, s-yg_{\leq}, m-yh_{\leq}, m-yh_{\leq}))$ if $g \geq 1, h \geq 1,$ and $k=0$.

\item \label{2_iii} $p(n, s, m, {f}_{=}, {g}_{=}, {h}_{=}, {k}_{=} )= \sum_{y}w(f,g,h; y)(p(n-yf, s-yg, m-yh, f_{=}, min\{m-yg, g-1\}_{\leq},m_{\leq}, m_{\leq})+p(n-yf, s-yg, m, min\{n-yf-1, f-1\}_{\leq}, s-yg_{\leq}, m_{\leq}, m_{\leq}))$ if $g \geq 1, h =0,$ and $k=0$.

\item \label{2_iv} $p(n, s, m, {f}_{=}, {g}_{=}, {h}_{=}, {k}_{=} )= \sum_{y}w(f,g,h; y)(p(n-yf, s, m-yh, f_{=}, g_{=}, min\{m-yh, h-1\}_{\leq}, m-yh_{\leq})+p(n-yf, s, m-yh, min\{n-yf-1, f-1\}_{\leq}, s_{\leq}, m-yh_{\leq}, m-yh_{\leq}))$ if $g =0, h \geq 1,$ and $k=0$.

\item \label{2_v} $p(n, s, m, {f}_{=}, {g}_{=}, {h}_{=}, {k}_{=} )= \sum_{y}w(f,g,h; y)(p(n-yf, s, m-yk, f_{=}, g_{=}, h_{=}, min\{m-yk, k-1\}_{\leq})+p(n-yf, s, m-yk, min\{n-yf-1, f-1\}_{\leq}, s_{\leq}, m-yk_{\leq}, m-yk_{\leq}))$ if $g =0, h =0,$ and $k \geq 1$.

\item \label{2_vi} $p(n, s, m, {f}_{=}, {g}_{=}, {h}_{=}, {k}_{=} )= \sum_{y}w(f,g,h; y)(p(n-yf, s-yg, m-yk, f_{=}, g_{=}, h_{=}, min\{m-yk, k-1\}_{\leq})+p(n-yf, s-yg, m-yk, f_{=}, min\{m-yg, g-1\}_{\leq},m-yk_{\leq}, m-yk_{\leq})+p(n-yf, s-yg, m-yk, min\{n-yf-1, f-1\}_{\leq}, s-yg_{\leq}, m-yk_{\leq}, m-yk_{\leq}))$ if $g \geq 1, h =0,$ and $k \geq 1$.

\item \label{2_vii} $p(n, s, m, {f}_{=}, {g}_{=}, {h}_{=}, {k}_{=} )= \sum_{y}w(f,g,h; y)(p(n-yf, s, m-y(h+k), f_{=}, g_{=}, h_{=}, min\{m-y(h+k), k-1\}_{\leq})+p(n-yf, s, m-y(h+k), f_{=}, g_{=}, min\{m-y(h+k), h-1\}_{\leq}, m-y(h+k)_{\leq})+p(n-yf, s, m-y(h+k), min\{n-yf-1, f-1\}_{\leq}, s_{\leq}, m-y(h+k)_{\leq}, m-y(h+k)_{\leq}))$ if $g =0, h \geq 1,$ and $k \geq 1$.

\item \label{2_viii} $p(n, s, m, {f}_{=}, {g}_{=}, {h}_{=}, {k}_{=} )= \sum_{y}w(f,g,h; y)(p(n-yf, s-yg, m-y(h+k), f_{=}, g_{=}, h_{=}, min\{m-y(h+k), k-1\}_{\leq})+p(n-yf, s-yg, m-y(h+k), f_{=}, g_{=}, min\{m-y(h+k), h-1\}_{\leq}, m-y(h+k)_{\leq})+p(n-yf, s-yg, m-y(h+k), f_{=}, min\{m-yg, g-1\}_{\leq}, m-y(h+k)_{\leq}, m-y(h+k)_{\leq})+p(n-yf, s-yg, m-y(h+k), min\{n-yf-1, f-1\}_{\leq}, s-yg_{\leq}, m-y(h+k)_{\leq}, m-y(h+k)_{\leq}))$ if $g \geq 1, h \geq 1,$ and $k \geq 1$.
\end{enumerate}
\end{lemma}
\begin{proof}
Let us consider a tree-like multigraph $P$ from the family $\mathcal{P}(n, s, m, {f}_{=}, {g}_{=}, {h}_{=}, {k}_{=} )$. Since $p(f,g, h, \,f{-}1_\le,\,g_\le,\,h_\le, h_\le)$ represent the number of tree-like multigraphs in the family $\mathcal{P}(f,g, h, \,f{-}1_\le,\,g_\le,\,h_\le, h_\le)$. By Lemma~\ref{lem_1}\ref{L_1_i} for the specific value of $y$ as $1 \leq y \leq \lfloor n-1 / f \rfloor$  with  $y \leq \lfloor s / g \rfloor$ when  $g \geq 1$ and $y \leq \lfloor m / (h+k) \rfloor$ when $h+k \geq 1 $, there are exactly $y$ descendent subgraphs $P_v$ for $v \in N(r_P)$ and $P_v \in \mathcal{P}(f,g, h, \,f{-}1_\le,\,g_\le,\,h_\le, h_\le)$. As  $\binom{p(f,g, h, \,f{-}1_\le,\,g_\le,\,h_\le, h_\le)+ y - 1}{y}$ denote the number of combinations with repetition of $y$ descendant subgraphs from the family 
$\mathcal{P}(f,g, h, \,f{-}1_\le,\,g_\le,\,h_\le, h_\le)$ for $v \in N(r_P)$. Further by Lemma~\ref{lem_1}\ref{L_1_ii} the residual subgraph $F$ of $P$ belongs to exactly one of the following families,\\

\noindent $\mathcal{P}(n-yf, s, m, min\{n-yf-1, f-1\}_{\leq}, s_{\leq}, m_{\leq}, m_{\leq})$ if $g=0, h=0,$ and $k=0$.
\vspace{10pt}

\noindent $\mathcal{P}(n-yf, s-yg, m-yh, f_{=}, g_{=}, min\{m-yh, h-1\}_{\leq}, m-yh_{\leq}) \cup \mathcal{P}(n-yf, s-yg, m-yh, f_{=}, min\{m-yg, g-1\}_{\leq},m-yh_{\leq}, m-yh_{\leq})\cup \mathcal{P}(n-yf, s-yg, m-yh, min\{n-yf-1, f-1\}_{\leq}, s-yg_{\leq}, m-yh_{\leq}, m-yh_{\leq})$ if $g \geq 1, h \geq 1,$ and $k=0$.

\vspace{10pt}

\noindent $\mathcal{P}(n-yf, s-yg, m-yh, f_{=}, min\{m-yg, g-1\}_{\leq},m_{\leq}, m_{\leq})\cup \mathcal{P}(n-yf, s-yg, m, min\{n-yf-1, f-1\}_{\leq}, s-yg_{\leq}, m_{\leq}, m_{\leq})$ if $g \geq 1, h =0,$ and $k=0$.
\vspace{10pt}

\noindent $\mathcal{P}(n-yf, s, m-yh, f_{=}, g_{=}, min\{m-yh, h-1\}_{\leq}, m-yh_{\leq})\cup \mathcal{P}(n-yf, s, m-yh, min\{n-yf-1, f-1\}_{\leq}, s_{\leq}, m-yh_{\leq}, m-yh_{\leq})$ if $g =0, h \geq 1,$ and $k=0$.
\vspace{10pt}

\noindent  $\mathcal{P}(n-yf, s, m-yk, f_{=}, g_{=}, h_{=}, min\{m-yk, k-1\}_{\leq})\cup \mathcal{P}(n-yf, s, m-yk, min\{n-yf-1, f-1\}_{\leq}, s_{\leq}, m-yk_{\leq}, m-yk_{\leq})$ if $g =0, h =0,$ and $k \geq 1$.

\vspace{10pt}

\noindent $\mathcal{P}(n-yf, s-yg, m-yk, f_{=}, g_{=}, h_{=}, min\{m-yk, k-1\}_{\leq})\cup \mathcal{P}(n-yf, s-yg, m-yk, f_{=}, min\{m-yg, g-1\}_{\leq},m-yk_{\leq}, m-yk_{\leq})\cup \mathcal{P}(n-yf, s-yg, m-yk, min\{n-yf-1, f-1\}_{\leq}, s-yg_{\leq}, m-yk_{\leq}, m-yk_{\leq})$ if $g \geq 1, h =0,$ and $k \geq 1$.
\vspace{10pt}

\noindent $\mathcal{P}(n-yf, s, m-y(h+k), f_{=}, g_{=}, h_{=}, min\{m-y(h+k), k-1\}_{\leq})\cup \mathcal{P}(n-yf, s, m-y(h+k), f_{=}, g_{=}, min\{m-y(h+k), h-1\}_{\leq}, m-y(h+k)_{\leq})\cup \mathcal{P}(n-yf, s, m-y(h+k), min\{n-yf-1, f-1\}_{\leq}, s_{\leq}, m-y(h+k)_{\leq}, m-y(h+k)_{\leq})$ if $g =0, h \geq 1,$ and $k \geq 1$.
\vspace{10pt}

\noindent $\mathcal{P}(n-yf, s-yg, m-y(h+k), f_{=}, g_{=}, h_{=}, min\{m-y(h+k), k-1\}_{\leq})\cup \mathcal{P}(n-yf, s-yg, m-y(h+k), f_{=}, g_{=}, min\{m-y(h+k), h-1\}_{\leq}, m-y(h+k)_{\leq})\cup \mathcal{P}(n-yf, s-yg, m-y(h+k), f_{=}, min\{m-yg, g-1\}_{\leq}, m-y(h+k)_{\leq}, m-y(h+k)_{\leq})\cup \mathcal{P}(n-yf, s-yg, m-y(h+k), min\{n-yf-1, f-1\}_{\leq}, s-yg_{\leq}, m-y(h+k)_{\leq}, m-y(h+k)_{\leq})$ if $g \geq 1, h \geq 1,$ and $k \geq 1$.\\

Note that in each case, the intersection of any two families of the residual trees will always result in an empty set.
By the sum rule of counting, the total number of tree-like multigraphs in the families of residual tree-like multigraphs is
\begin{enumerate}[label=(\alph*), ref  = (\alph*), font=\upshape]
\item \label{a} $p(n-yf, s, m, min\{n-yf-1, f-1\}_{\leq}, s_{\leq}, m_{\leq}, m_{\leq})$ if $g=0, h=0,$ and $k=0$.

\item \label{b} $ p(n-yf, s-yg, m-yh, f_{=}, g_{=}, min\{m-yh, h-1\}_{\leq}, m-yh_{\leq})+ p(n-yf, s-yg, m-yh, f_{=}, min\{m-yg, g-1\}_{\leq},m-yh_{\leq}, m-yh_{\leq})+ p(n-yf, s-yg, m-yh, min\{n-yf-1, f-1\}_{\leq}, s-yg_{\leq}, m-yh_{\leq}, m-yh_{\leq}))$ if $g \geq 1, h \geq 1,$ and $k=0$.

\item \label{c} $p(n-yf, s-yg, m-yh, f_{=}, min\{m-yg, g-1\}_{\leq},m_{\leq}, m_{\leq})+p(n-yf, s-yg, m, min\{n-yf-1, f-1\}_{\leq}, s-yg_{\leq}, m_{\leq}, m_{\leq})$ if $g \geq 1, h =0,$ and $k=0$.

\item \label{d} $p(n-yf, s, m-yh, f_{=}, g_{=}, min\{m-yh, h-1\}_{\leq}, m-yh_{\leq})+p(n-yf, s, m-yh, min\{n-yf-1, f-1\}_{\leq}, s_{\leq}, m-yh_{\leq}, m-yh_{\leq})$ if $g =0, h \geq 1,$ and $k=0$.

\item \label{e} $p(n-yf, s, m-yk, f_{=}, g_{=}, h_{=}, min\{m-yk, k-1\}_{\leq})+p(n-yf, s, m-yk, min\{n-yf-1, f-1\}_{\leq}, s_{\leq}, m-yk_{\leq}, m-yk_{\leq})$ if $g =0, h =0,$ and $k \geq 1$.

\item \label{f} $p(n-yf, s-yg, m-yk, f_{=}, g_{=}, h_{=}, min\{m-yk, k-1\}_{\leq})+p(n-yf, s-yg, m-yk, f_{=}, min\{m-yg, g-1\}_{\leq},m-yk_{\leq}, m-yk_{\leq})+p(n-yf, s-yg, m-yk, min\{n-yf-1, f-1\}_{\leq}, s-yg_{\leq}, m-yk_{\leq}, m-yk_{\leq})$ if $g \geq 1, h =0,$ and $k \geq 1$.

\item \label{g} $p(n-yf, s, m-y(h+k), f_{=}, g_{=}, h_{=}, min\{m-y(h+k), k-1\}_{\leq})+p(n-yf, s, m-y(h+k), f_{=}, g_{=}, min\{m-y(h+k), h-1\}_{\leq}, m-y(h+k)_{\leq})+p(n-yf, s, m-y(h+k), min\{n-yf-1, f-1\}_{\leq}, s_{\leq}, m-y(h+k)_{\leq}, m-y(h+k)_{\leq})$ if $g =0, h \geq 1,$ and $k \geq 1$.

\item \label{h} $p(n-yf, s-yg, m-y(h+k), f_{=}, g_{=}, h_{=}, min\{m-y(h+k), k-1\}_{\leq})+p(n-yf, s-yg, m-y(h+k), f_{=}, g_{=}, min\{m-y(h+k), h-1\}_{\leq}, m-y(h+k)_{\leq})+p(n-yf, s-yg, m-y(h+k), f_{=}, min\{m-yg, g-1\}_{\leq}, m-y(h+k)_{\leq}, m-y(h+k)_{\leq})+p(n-yf, s-yg, m-y(h+k), min\{n-yf-1, f-1\}_{\leq}, s-yg_{\leq}, m-y(h+k)_{\leq}, m-y(h+k)_{\leq})$ if $g \geq 1, h \geq 1,$ and $k \geq 1$.
\end{enumerate}
 This implies that for a fixed integer $y$ in the range given in the statement, and by the product rule of counting, the number of residual multigraphs $F$ in the family $\mathcal{P}(n, s, m, {f}_{=}, {g}_{=}, {h}_{=}, {k}_{=} )$ with exactly $y$ descendent subgraphs $F_v \in \mathcal{P}(f,g, h, \,f{-}1_\le,\,g_\le,\,h_\le, h_\le)$, for $v \in P{r_F},$ are
 \begin{enumerate}[label=(\alph*), ref  = (\alph*), font=\upshape]
 \item \label{a} $w(f,g,h; y)(p(n-yf, s, m, min\{n-yf-1, f-1\}_{\leq}, s_{\leq}, m_{\leq}, m_{\leq}))$ if $g=0, h=0,$ and $k=0$.

\item \label{b} $w(f,g,h; y)( p(n-yf, s-yg, m-yh, f_{=}, g_{=}, min\{m-yh, h-1\}_{\leq}, m-yh_{\leq})+ p(n-yf, s-yg, m-yh, f_{=}, min\{m-yg, g-1\}_{\leq},m-yh_{\leq}, m-yh_{\leq})+ p(n-yf, s-yg, m-yh, min\{n-yf-1, f-1\}_{\leq}, s-yg_{\leq}, m-yh_{\leq}, m-yh_{\leq}))$ if $g \geq 1, h \geq 1,$ and $k=0$.

\item \label{c} $w(f,g,h; y)(p(n-yf, s-yg, m-yh, f_{=}, min\{m-yg, g-1\}_{\leq},m_{\leq}, m_{\leq})+p(n-yf, s-yg, m, min\{n-yf-1, f-1\}_{\leq}, s-yg_{\leq}, m_{\leq}, m_{\leq}))$ if $g \geq 1, h =0,$ and $k=0$.

\item \label{d} $w(f,g,h; y)(p(n-yf, s, m-yh, f_{=}, g_{=}, min\{m-yh, h-1\}_{\leq}, m-yh_{\leq})+p(n-yf, s, m-yh, min\{n-yf-1, f-1\}_{\leq}, s_{\leq}, m-yh_{\leq}, m-yh_{\leq}))$ if $g =0, h \geq 1,$ and $k=0$.

\item \label{e} $w(f,g,h; y)(p(n-yf, s, m-yk, f_{=}, g_{=}, h_{=}, min\{m-yk, k-1\}_{\leq})+p(n-yf, s, m-yk, min\{n-yf-1, f-1\}_{\leq}, s_{\leq}, m-yk_{\leq}, m-yk_{\leq}))$ if $g =0, h =0,$ and $k \geq 1$.

\item \label{f} $w(f,g,h; y)(p(n-yf, s-yg, m-yk, f_{=}, g_{=}, h_{=}, min\{m-yk, k-1\}_{\leq})+p(n-yf, s-yg, m-yk, f_{=}, min\{m-yg, g-1\}_{\leq},m-yk_{\leq}, m-yk_{\leq})+p(n-yf, s-yg, m-yk, min\{n-yf-1, f-1\}_{\leq}, s-yg_{\leq}, m-yk_{\leq}, m-yk_{\leq}))$ if $g \geq 1, h =0,$ and $k \geq 1$.

\item \label{g} $w(f,g,h; y)(p(n-yf, s, m-y(h+k), f_{=}, g_{=}, h_{=}, min\{m-y(h+k), k-1\}_{\leq})+p(n-yf, s, m-y(h+k), f_{=}, g_{=}, min\{m-y(h+k), h-1\}_{\leq}, m-y(h+k)_{\leq})+p(n-yf, s, m-y(h+k), min\{n-yf-1, f-1\}_{\leq}, s_{\leq}, m-y(h+k)_{\leq}, m-y(h+k)_{\leq}))$ if $g =0, h \geq 1,$ and $k \geq 1$.

\item \label{h} $w(f,g,h; y)(p(n-yf, s-yg, m-y(h+k), f_{=}, g_{=}, h_{=}, min\{m-y(h+k), k-1\}_{\leq})+p(n-yf, s-yg, m-y(h+k), f_{=}, g_{=}, min\{m-y(h+k), h-1\}_{\leq}, m-y(h+k)_{\leq})+p(n-yf, s-yg, m-y(h+k), f_{=}, min\{m-yg, g-1\}_{\leq}, m-y(h+k)_{\leq}, m-y(h+k)_{\leq})+p(n-yf, s-yg, m-y(h+k), min\{n-yf-1, f-1\}_{\leq}, s-yg_{\leq}, m-y(h+k)_{\leq}, m-y(h+k)_{\leq}))$ if $g \geq 1, h \geq 1,$ and $k \geq 1$.
 \end{enumerate}
 Note that in the case $n=1, s \geq 1$ and $m \geq 1$, $p(1, s, m)=0$, whereas for $n=2, s \geq 1$ and $m \geq 1$, we have $p(2,s,m)=1$. Similarly, for $f=1$, $s=0$ and $m=0$, we have $1 \leq y \leq n-1$, then by Lemma~\ref{lem_4}\ref{4_ii}, it holds that $p(n-y, 0, 0, \,0_\le,\,0_\le,\,0_\le, 0_\le)=1$ if $y=n-1$ and $p(n-y, 0, 0, \,0_\le,\,0_\le,\,0_\le, 0_\le)=0$ if $1 \leq y \leq n-2$. This implies $P \in \mathcal{P}(n, 0, 0, {1}_{=}, {0}_{=}, {0}_{=}, {0}_{=} )$ has exactly $y=n-1$ descendent subgraphs. $P_v \in \mathcal{P}(1, 0, 0, \,0_\le,\,0_\le,\,0_\le, 0_\le)$ for $v \in N(r_P)$. However observe that for each integer $f \geq 2$ or $h+k=1$, $g=1$ and $y$ satisfying the condition given in the statement, there exist at least one tree-like multigraph $P \in \mathcal{P}(n, s, m, {f}_{=}, {g}_{=}, {h}_{=}, {k}_{=} )$ such that $P$ has exactly $y$ descendent subgraphs $P_v \in \mathcal{P}(f,g, h, \,f{-}1_\le,\,g_\le,\,h_\le, h_\le)$ for $v \in N(r_P)$. This and cases~\ref{a}-\ref{h} implies the statement statements~\ref{2_i}-\ref{2_viii} respectively.
\end{proof}
\begin{theorem}
\label{theorem_1}
For any eight integers, $n \geq 3$, $f \geq 1$, $s \geq g \geq 0$, $m \geq h+k \geq 0$ and $y \geq 0$, such that $1 \leq y \leq \lfloor n-1 / f \rfloor$  with  $y \leq \lfloor s / g \rfloor$ when  $g \geq 1$ also $y \leq \lfloor m / (h+k) \rfloor$ when $h+k \geq 1 $and $0 \leq k \leq m$, it holds that

\begin{enumerate}[label=(\roman*), ref  = (\roman*), font=\upshape]
\item \label{theorem_1_i} $p(n, s, m, {f}_{=}, {g}_{=}, {h}_{=}, {k}_{=} )= \sum_{y}w(f,g,h; y)((p(f,g, h, \,f{-}1_\le,\,g_\le,\,h_\le, h_\le)+y-1)/y)(p(n-yf, s, m, min\{n-yf-1, f-1\}_{\leq}, s_{\leq}, m_{\leq}, m_{\leq}))$ if $g=0, h=0,$ and $k=0$.

\item \label{theorem_1_ii} $p(n, s, m, {f}_{=}, {g}_{=}, {h}_{=}, {k}_{=} )= \sum_{y}w(f,g,h; y)((p(f,g, h, \,f{-}1_\le,\,g_\le,\,h_\le, h_\le)+y-1)/y)( p(n-yf, s-yg, m-yh, f_{=}, g_{=}, min\{m-yh, h-1\}_{\leq}, m-yh_{\leq})+ p(n-yf, s-yg, m-yh, f_{=}, min\{m-yg, g-1\}_{\leq},m-yh_{\leq}, m-yh_{\leq})+ p(n-yf, s-yg, m-yh, min\{n-yf-1, f-1\}_{\leq}, s-yg_{\leq}, m-yh_{\leq}, m-yh_{\leq}))$ if $g \geq 1, h \geq 1,$ and $k=0$.

\item \label{theorem_1_iii} $p(n, s, m, {f}_{=}, {g}_{=}, {h}_{=}, {k}_{=} )= \sum_{y}w(f,g,h; y)((p(f,g, h, \,f{-}1_\le,\,g_\le,\,h_\le, h_\le)+y-1)/y)(p(n-yf, s-yg, m-yh, f_{=}, min\{m-yg, g-1\}_{\leq},m_{\leq}, m_{\leq})+p(n-yf, s-yg, m, min\{n-yf-1, f-1\}_{\leq}, s-yg_{\leq}, m_{\leq}, m_{\leq}))$ if $g \geq 1, h =0,$ and $k=0$.

\item \label{theorem_1_iv} $p(n, s, m, {f}_{=}, {g}_{=}, {h}_{=}, {k}_{=} )= \sum_{y}w(f,g,h; y)((p(f,g, h, \,f{-}1_\le,\,g_\le,\,h_\le, h_\le)+y-1)/y)(p(n-yf, s, m-yh, f_{=}, g_{=}, min\{m-yh, h-1\}_{\leq}, m-yh_{\leq})+p(n-yf, s, m-yh, min\{n-yf-1, f-1\}_{\leq}, s_{\leq}, m-yh_{\leq}, m-yh_{\leq}))$ if $g =0, h \geq 1,$ and $k=0$.

\item \label{theorem_1_v} $p(n, s, m, {f}_{=}, {g}_{=}, {h}_{=}, {k}_{=} )= \sum_{y}w(f,g,h; y)((p(f,g, h, \,f{-}1_\le,\,g_\le,\,h_\le, h_\le)+y-1)/y)(p(n-yf, s, m-yk, f_{=}, g_{=}, h_{=}, min\{m-yk, k-1\}_{\leq})+p(n-yf, s, m-yk, min\{n-yf-1, f-1\}_{\leq}, s_{\leq}, m-yk_{\leq}, m-yk_{\leq}))$ if $g =0, h =0,$ and $k \geq 1$.

\item \label{theorem_1_vi} $p(n, s, m, {f}_{=}, {g}_{=}, {h}_{=}, {k}_{=} )= \sum_{y}w(f,g,h; y)((p(f,g, h, \,f{-}1_\le,\,g_\le,\,h_\le, h_\le)+y-1)/y)(p(n-yf, s-yg, m-yk, f_{=}, g_{=}, h_{=}, min\{m-yk, k-1\}_{\leq})+p(n-yf, s-yg, m-yk, f_{=}, min\{m-yg, g-1\}_{\leq},m-yk_{\leq}, m-yk_{\leq})+p(n-yf, s-yg, m-yk, min\{n-yf-1, f-1\}_{\leq}, s-yg_{\leq}, m-yk_{\leq}, m-yk_{\leq}))$ if $g \geq 1, h =0,$ and $k \geq 1$.

\item \label{theorem_1_vii} $p(n, s, m, {f}_{=}, {g}_{=}, {h}_{=}, {k}_{=} )= \sum_{y}w(f,g,h; y)((p(f,g, h, \,f{-}1_\le,\,g_\le,\,h_\le, h_\le)+y-1)/y)(p(n-yf, s, m-y(h+k), f_{=}, g_{=}, h_{=}, min\{m-y(h+k), k-1\}_{\leq})+p(n-yf, s, m-y(h+k), f_{=}, g_{=}, min\{m-y(h+k), h-1\}_{\leq}, m-y(h+k)_{\leq})+p(n-yf, s, m-y(h+k), min\{n-yf-1, f-1\}_{\leq}, s_{\leq}, m-y(h+k)_{\leq}, m-y(h+k)_{\leq}))$ if $g =0, h \geq 1,$ and $k \geq 1$.

\item \label{theorem_1_viii} $p(n, s, m, {f}_{=}, {g}_{=}, {h}_{=}, {k}_{=} )= \sum_{y}w(f,g,h; y)((p(f,g, h, \,f{-}1_\le,\,g_\le,\,h_\le, h_\le)+y-1)/y)(p(n-yf, s-yg, m-y(h+k), f_{=}, g_{=}, h_{=}, min\{m-y(h+k), k-1\}_{\leq})+p(n-yf, s-yg, m-y(h+k), f_{=}, g_{=}, min\{m-y(h+k), h-1\}_{\leq}, m-y(h+k)_{\leq})+p(n-yf, s-yg, m-y(h+k), f_{=}, min\{m-yg, g-1\}_{\leq}, m-y(h+k)_{\leq}, m-y(h+k)_{\leq})+p(n-yf, s-yg, m-y(h+k), min\{n-yf-1, f-1\}_{\leq}, s-yg_{\leq}, m-y(h+k)_{\leq}, m-y(h+k)_{\leq}))$ if $g \geq 1, h \geq 1,$ and $k \geq 1$.
\end{enumerate}
\end{theorem}
\begin{proof}
The proof of Theorem\ref{theorem_1} follows from Lemma~\ref{lem_2}. Furthermore, it holds that 
\begin{align*}
 w(f, g, h ; y) = &\ \frac{(p(f,g, h, \,f{-}1_\le,\,g_\le,\,h_\le, h_\le) + y - 1)!}{(p(f,g, h, \,f{-}1_\le,\,g_\le,\,h_\le, h_\le) - 1)!y!} \\
=& \frac{(p(f,g, h, \,f{-}1_\le,\,g_\le,\,h_\le, h_\le) + y -2)!}{p(f,g, h, \,f{-}1_\le,\,g_\le,\,h_\le, h_\le) - 1)!(y - 1)!} \times
\vspace{2mm} \\
& \hspace{4.2cm} \frac{(p(f,g, h, \,f{-}1_\le,\,g_\le,\,h_\le, h_\le) + y - 1)}{p}\\
  =& w(f, g, h; y - 1) \times \frac{p(f,g, h, \,f{-}1_\le,\,g_\le,\,h_\le, h_\le) + y - 1)}{y}.
 \end{align*}
\end{proof}
In Lemma~\ref{lem_3}, we discuss the recursive relations for $p(n, s, m, {f}_{\leq}, {g}_{\leq}, {h}_{\leq}, {k}_{\leq} )$, $p(n, s, m, {f}_{=}, {g}_{\leq}, {h}_{\leq}, {k}_{\leq} )$, $p(n, s, m, {f}_{=}, {g}_{=}, {h}_{\leq}, {k}_{\leq} )$ and   $p(n, s, m, {f}_{=}, {g}_{=}, {h}_{=}, {k}_{\leq} )$.
\begin{lemma}
\label{lem_3}
For any seven integers $n-1 \geq f \geq 0$, $s\geq g \geq 0$ and $m \geq h+k \geq 0$ it holds that 
\begin{enumerate}[label=\textnormal{(\roman*)}, ref=(\roman*), font=\upshape]

\item \label{3_i} 
$p(n, s, m, {f}_{\leq}, {g}_{\leq}, {h}_{\leq}, {k}_{\leq} )= p(n, s, m, {0}_{=}, {g}_{\leq}, {h}_{\leq}, {k}_{\leq} ) \text{ if } f = 0,$ ;

\item \label{3_ii} 
$p(n, s, m, {f}_{\leq}, {g}_{\leq}, {h}_{\leq}, {k}_{\leq} )= 
p(n, s, m, {f-1}_{\leq}, {g}_{\leq}, {h}_{\leq}, {k}_{\leq} )+ p(n, s, m, {f}_{=}, {g}_{\leq}, {h}_{\leq}, {k}_{\leq} ) 
\text { if } f \geq 1$;

\item \label{3_iii} 
$p(n, s, m, {f}_{=}, {g}_{\leq}, {h}_{\leq}, {k}_{\leq} )= 
p(n, s, m, {f}_{=}, {0}_{=}, {h}_{\leq}, {k}_{\leq} ) \text{ if } g = 0$; 

\item \label{3_iv} 
$p(n, s, m, {f}_{=}, {g}_{\leq}, {h}_{\leq}, {k}_{\leq} )= 
p(n, s, m, {f}_{=}, {g-1}_{\leq}, {h}_{\leq}, {k}_{\leq} )+ p(n, s, m, {f}_{=}, {g}_{=}, {h}_{\leq}, {k}_{\leq} ) \text { if } g \geq 1$;

\item \label{3_v}
$p(n, s, m, {f}_{=}, {g}_{=}, {h}_{\leq}, {k}_{\leq} )= 
p(n, s, m, {f}_{=}, {g}_{=}, {0}_{=}, {k}_{\leq} ) \text{ if } h = 0,$

 \item \label{3_vi}
$p(n, s, m, {f}_{=}, {g}_{=}, {h}_{\leq}, {k}_{\leq} )= 
p(n, s, m, {f}_{=}, {g}_{=}, {h-1}_{\leq}, {k}_{\leq} )+ p(n, s, m, {f}_{=}, {g}_{=}, {h}_{=}, {k}_{\leq} ) 
\text { if } h \geq 1$;

\item \label{3_vii}
$p(n, s, m, {f}_{=}, {g}_{=}, {h}_{=}, {k}_{\leq} )= 
p(n, s, m, {f}_{=}, {g}_{=}, {h}_{=}, {0}_{=} ) \text{ if } k = 0$; and 

\item \label{3_viii}
$p(n, s, m, {f}_{=}, {g}_{=}, {h}_{=}, {k}_{\leq} )= 
p(n, s, m, {f}_{=}, {g}_{=}, {h}_{\leq}, {k-1}_{\leq} )+p(n, s, m, {f}_{=}, {g}_{=}, {h}_{=}, {k}_{=} ) 
\text { if } k \geq 1$. 
\end{enumerate}
\end{lemma}
\begin{proof}
By Eq.~(\ref{equ1}), the case~\ref{3_i} follows. Similarly,  by Eq.~(\ref{equ1p}), and the fact $f \geq 1$, it holds that $\mathcal{P}(n, s, m, {f-1}_{\leq}, {g}_{\leq}, {h}_{\leq}, {k}_{\leq} )\cap 
\mathcal{P}(n, s, m, {f}_{=}, {g}_{\leq}, {h}_{\leq}, {k}_{\leq} ) = \emptyset$, case~\ref{3_ii} follows. By Eq.~(\ref{equ2}), the case~\ref{3_iii} holds. Similarly, by Eq.~(\ref{equ2}), and the fact $f \geq 1$, it holds that $\mathcal{P}(n, s, m, {f}_{=}, {g-1}_{\leq}, {h}_{\leq}, {k}_{\leq} )\cap 
\mathcal{P}(n, s, m, {f}_{=}, {g}_{=}, {h}_{\leq}, {k}_{\leq} ) = \emptyset$, the case~\ref{3_iv} holds. By Eq.~(\ref{equ3}), the case~\ref{3_v} follows. Similarly,  by Eq.~(\ref{equ3p}), and the fact $h \geq 1$, it holds that $\mathcal{P}(n, s, m, {f}_{=}, {g}_{=}, {h-1}_{\leq}, {k}_{\leq} )\cap 
\mathcal{P}(n, s, m, {f}_{=}, {g}_{=}, {h}_{=}, {k}_{\leq} ) = \emptyset$, case~\ref{3_vi} follows. And by Eq.~(\ref{equ4}), the case~\ref{3_vii} follows. Similarly,  by Eq.~(\ref{equ4p}), and the fact $k \geq 1$, it holds that $\mathcal{P}(n, s, m, {f}_{=}, {g}_{=}, {h}_{\leq}, {k-1}_{\leq} )\cap 
\mathcal{P}(n, s, m, {f}_{=}, {g}_{=}, {h}_{=}, {k}_{=} ) = \emptyset$, case~\ref{3_viii} follows.
\end{proof}

\subsection{Initial conditions}
In Lemma~\ref{lem_4}, we discuss the initial conditions that are necessary to design a DP algorithm. 
\begin{lemma}
\label{lem_4}
For any seven integers $n-1 \geq f \geq 0$, $s\geq g \geq 0$ and $m \geq h+k \geq 0$ it holds that 
\begin{enumerate}[label=\textnormal{(\roman*)}, ref=(\roman*), font=\upshape]

\item \label{4_i} 
$p(n, s, m, {0}_{=}, {g}_{=}, {h}_{=}, {k}_{=})= 1$ if $n=1$, $s=0$ and $m=0$. Further, $p(n, s, m, {0}_{\leq}, {g}_{\leq}, {h}_{\leq}, {k}_{\leq})= 0$  if   $n=1$, $s \geq 1$ or $m \geq 1$ or $n \geq 2$.

\item \label{4_ii} 
$p(n, s, m, {0}_{=}, {g}_{\leq}, {h}_{\leq}, {k}_{\leq})= p(n, s, m, {0}_{\leq}, {g}_{\leq}, {h}_{\leq}, {k}_{\leq})=1$ if $n=1$, $s=0$ and $m=0$ and  $p(n, s, m, {0}_{=}, {g}_{\leq}, {h}_{\leq}, {k}_{\leq})= p(n, s, m, {0}_{\leq}, {g}_{\leq}, {h}_{\leq}, {k}_{\leq})=0$ if $n \geq 2$.

\item \label{4_iii}
$p(n, s, m, {1}_{=}, {g}_{=}, {h}_{=}, {k}_{\leq})=0$ if $h \geq 1$ similarly, $p(n, s, m, {1}_{=}, {0}_{=}, {0}_{=}, {0}_{\leq})=0$ if $s \geq 1$ and $m \geq 1$.

\item \label{4_iv} 
$p(n, s, m, {1}_{=}, {g}_{=}, {0}_{=}, {k}_{=})=p(n, s, m, {1}_{=}, {g}_{=}, {0}_{=}, {k}_{\leq})=p(n, s, m, {1}_{=}, {g}_{=}, {h}_{\leq}, {k}_{\leq})=p(n, s, m, {1}_{=}, {g}_{\leq}, {h}_{\leq}, {k}_{\leq})=p(n, s, m, {1}_{\leq}, {g}_{\leq}, {h}_{\leq}, {k}_{\leq})=1$ if $n=2$ and $k=m$ or $n \geq 3$, $g=0$ and $k=0$;

\item \label{4_v} 
$p(n, s, m, {f}_{=}, {g}_{=}, {h}_{=}, {k}_{=})=0$ if $g > s$ and $h+k > m$;

\item \label{4_vi} 
$p(n, s, m, {f}_{=}, {g}_{=}, {h}_{=}, {k}_{=})=p(n, s, m, {f}_{=}, {g}_{=}, {h}_{=}, {k}_{\leq})=p(n, s, m, {f}_{=}, {g}_{=}, {h}_{\leq}, {k}_{\leq})=p(n, s, m, {f}_{=}, {g}_{\leq}, {h}_{\leq}, {k}_{\leq})=0$ if $f\geq n$.

\end{enumerate}
\end{lemma}
\begin{proof}
\begin{enumerate}[label=\textnormal{(\roman*)}, ref=(\roman*), font=\upshape]
\item A multigraph $P$ with $\mathrm{Max}_v(P) = 0$ exists if and only if $|V(P)| = 1$, $\mathrm{Max}_s(P) = 0$, $\mathrm{Max}_m(P) = 0$, and $\mathrm{Max}_\ell(P) = 0$. This occurs only when $n = 1$, $s = 0$, $m = 0$, $g = 0$, $h = 0$, and $k = 0$. No such multigraph exists if $s \geq 1$ or $m \geq 1$ (a single vertex cannot have multiple edges), or if $n \geq 2$.

\item It follows from Lemma~\ref{lem_3}~\ref{3_i} and~\ref{3_ii}. A multigraph with $\mathrm{Max}_v(P) \leq 0$ exists if and only if $n = 1$, $s = 0$, and $m = 0$, for $n \geq 2$, at least one child must exist, contradicting $\mathrm{Max}_v(P) \leq 0$.
.

\item When $f = 1$, descendant subgraphs have size $1$. A single vertex cannot contain multiple edges within itself, so $h$ must be $0$. Hence, $p(n, s, m, {1}_{=}, {g}_{=}, {h}_{=}, {k}_{\leq}) = 0$ if $h \geq 1$. For the second part, when $g = 0$ and $h = 0$, if $s \geq 1$ or $m \geq 1$, there is no valid base configuration to accommodate these features with size-$1$ descendant subgraphs.

\item  When $n = 2$, there is exactly one child forming a descendant subgraph of size $1$. With $k = m$, all multiple edges are between the root and its single child, creating exactly one valid configuration. When $n \geq 3$ and $g = 0$ and $k = 0$, the structure is a root with $n - 1$ leaf children, having no self-loops in descendant subgraphs and no multiple edges connected to the root, again creating exactly one configuration. The equalities with upper bounds follow because the given constraints enforce unique configurations.
 
\item The maximum number of self-loops in descendant subgraphs cannot exceed the total number of self-loops, i.e., $g \leq s$. Similarly, the sum of multiple edges within descendant subgraphs and between the root and its children cannot exceed the total number of multiple edges, i.e., $h + k \leq m$. If $g > s$ or $h + k > m$, the specified requirements exceed the available resources, making such configurations impossible.

\item The maximum size of any descendant subgraph can be at most $n -1$ vertices, since the root occupies at least one vertex. If $f \geq n$, no such multigraph exists. The equalities follow directly from this definition.

\end{enumerate}
\end{proof}
\begin{theorem}
\label{theorem_2}
For any three integers $n \geq 1$, $s \geq 0$, and $m \geq 0$, the number of non-isomorphic rooted tree-like multigraphs with $n$ vertices and $s$ self-loops and $m$ multiple edges can be determined.
\end{theorem}
\section{Conclusion} \label{conclusion}
In this work, a generalized framework has been developed for counting non-isomorphic tree-like multigraphs that allow self-loops and multiple edges simultaneously. The proposed dynamic programming formulation unifies two previously independent approaches, extending the scope of graph enumeration to a broader class of combinatorial structures. By defining canonical rooted representations and deriving recursive relations for graph composition, the method ensures that every structure is counted exactly once, eliminating the need for explicit isomorphism checking or exhaustive enumeration. This formulation contributes to the theoretical understanding of graph enumeration by establishing a systematic way to represent and count complex multigraph configurations. It also provides a foundation for studying growth patterns, structural limits, and combinatorial properties of graph families that incorporate both self-connections and repeated edges. Beyond its mathematical value, the framework has potential applications in areas such as chemical graph theory, molecular topology, and the analysis of hierarchical networks where such structures frequently appear.
Future work may focus on extending this framework to labeled or weighted multigraphs, incorporating degree constraints. These extensions would further enhance the applicability of the model to real-world systems while deepening the theoretical connections between graph combinatorics and applied structural analysis.


\begin{thebibliography}{99}
\bibitem[Blum and Reymond(2009)]{Blum2009}
Blum, L. C.; Reymond, J.-L. 970 Million Druglike Small Molecules for Virtual Screening in the Chemical Universe Database GDB-13. \emph{J. Am. Chem. Soc.} \textbf{2009}, \emph{131}, 8732--8733. Available online: \url{https://doi.org/10.1021/ja902302h}.

\bibitem[Lim~et~al.(2020)]{Lim2020}
Lim, J.; Hwang, S.-Y.; Moon, S.; Kim, S.; Kim, W. Y. Scaffold-Based Molecular Design with a Graph Generative Model. \emph{Chem. Sci.} \textbf{2020}, \emph{11}, 1153--1164. Available online: \url{https://doi.org/10.1039/C9SC04503A}.

\bibitem[Meringer and Schymanski(2013)]{Meringer2013}
Meringer, M.; Schymanski, E. L. Small Molecule Identification with MOLGEN and Mass Spectrometry. \emph{Metabolites} \textbf{2013}, \emph{3}, 440--462. Available online: \url{https://doi.org/10.3390/metabo3020440}.

\bibitem[Benecke~et~al.(1995)]{Benecke1995}
Benecke, C.; Grund, R.; Hohberger, R.; Kerber, A.; Laue, R.; Wieland, T. MOLGEN+, a Generator of Connectivity Isomers and Stereoisomers for Molecular Structure Elucidation. \emph{Anal. Chim. Acta} \textbf{1995}, \emph{314}, 141--147. Available online: \url{https://doi.org/10.1016/0003-2670(95)00291-7}.

\bibitem[Peironcely~et~al.(2012)]{Peironcely2012}
Peironcely, J. E.; Rojas-Chert{\'o}, M.; Fichera, D.; Reijmers, T.; Coulier, L.; Faulon, J.-L.; Hankemeier, T. OMG: Open Molecule Generator. \emph{J. Cheminform.} \textbf{2012}, \emph{4}, 21. Available online: \url{https://doi.org/10.1186/1758-2946-4-21}.

\bibitem[Batool~et~al.(2024)]{Batool2024}
Batool, M.; Azam, N. A.; Khan, M. A Unified Approach to Inferring Chemical Compounds with the Desired Aqueous Solubility. \emph{J. Cheminform.} \textbf{2024}, \emph{16}, 66. Available online: \url{https://doi.org/10.1186/s13321-025-00966-w}.

\bibitem[Cao~et~al.(2024)]{Cao2024}
Cao, P.-Y.; He, Y.; Cui, M.-Y.; Zhang, X.-M.; Zhang, Q.; Zhang, H.-Y. Group Graph: A Molecular Graph Representation with Enhanced Performance, Efficiency and Interpretability. \emph{J. Cheminform.} \textbf{2024}, \emph{16}, 133. Available online: \url{http://dx.doi.org/10.1186/s13321-024-00933-x}.

\bibitem[Tezuka and Oike(2002)]{Tezuka2002}
Tezuka Y.; Oike H. Topological polymer chemistry. \emph{Prog. Polym. Sci.} \textbf{2002}, \emph{27}(6), 1069–1122. \href{https://doi.org/10.1016/S0079-6700(02)00009-6}{https://doi.org/10.1016/S0079-6700(02)00009-6}.

\bibitem[Galina and Sysło(1988)]{Galina1988}
Galina H.; Sysło M. M. Some applications of graph theory to the study of polymer configuration. \emph{Discrete Appl. Math.} \textbf{1988}, \emph{19}(1-3), 167–176. \href{https://doi.org/10.1016/0166-218X(88)90012-1}{https://doi.org/10.1016/0166-218X(88)90012-1}.

\bibitem[Azam~et~al.(2020b)]{Azam2020}
Azam, N. A.; Shurbevski, A.; Nagamochi, H. An Efficient Algorithm to Count Tree-Like Graphs with a Given Number of Vertices and Self-Loops. \emph{Entropy} \textbf{2020}, \emph{22}, 923. Available online: \url{https://doi.org/10.3390/e22090923}.

\bibitem[Ilyas et al.(2025)]{Azam2024}
Ilyas M.; Hayat S.; Azam N. A. Counting Tree-Like Multigraphs with a Given Number of Vertices and Multiple Edges. \emph{arXiv preprint} arXiv:2502.05529 (2025). \href{https://doi.org/10.48550/arXiv.2502.05529}{https://doi.org/10.48550/arXiv.2502.05529}.

\bibitem[P{\'o}lya and Read(2012)]{Polya2012}
P{\'o}lya, G.; Read, R. C. \emph{Combinatorial Enumeration of Groups, Graphs, and Chemical Compounds}; Springer Science \& Business Media: New York, USA, 2012. Available online: \url{https://doi.org/10.1007/978-1-4612-4664-0}.

\bibitem[P{\'o}lya(1937)]{Polya1937}
P{\'o}lya, G. Kombinatorische Anzahlbestimmungen f{\"u}r Gruppen, Graphen und chemische Verbindungen. \textbf{1937}. Available online: \url{https://doi.org/10.1007/BF02546665}.

\bibitem[McKay and Piperno(2014)]{McKay2014}
McKay, B. D.; Piperno, A. Practical Graph Isomorphism, II. \emph{J. Symb. Comput.} \textbf{2014}, \emph{60}, 94--112. Available online: \url{https://doi.org/10.1016/j.jsc.2013.09.003}.

\bibitem[Akutsu and Fukagawa(2007)]{Akutsu2007}
Akutsu, T.; Fukagawa, D. Inferring a Chemical Structure from a Feature Vector Based on Frequency of Labeled Paths and Small Fragments. In \emph{Proceedings of the 5th Asia-Pacific Bioinformatics Conference}, 2007; pp. 165--174. Available online: \url{https://doi.org/10.1142/9781860947995_0019}.

\bibitem[Fujiwara~et~al.(2008)]{Fujiwara2008}
Fujiwara, H.; Wang, J.; Zhao, L.; Nagamochi, H.; Akutsu, T. Enumerating Tree-Like Chemical Graphs with Given Path Frequency. \emph{J. Chem. Inf. Model.} \textbf{2008}, \emph{48}, 1345--1357. Available online: \url{https://doi.org/10.1021/ci700385a}.

\bibitem[Nakano and Uno(2003)]{Nakano2003}
Nakano, S.; Uno, T. Efficient Generation of Rooted Trees. \emph{National Institute for Informatics (Japan), Technical Report NII-2003-005E} \textbf{2003}, \emph{8}, 4--63.

\bibitem[Nakano and Uno(2005)]{Nakano2005}
Nakano, S.; Uno, T. Generating Colored Trees. In \emph{Graph-Theoretic Concepts in Computer Science: 31st International Workshop, WG 2005, Metz, France, June 23–25, 2005, Revised Selected Papers}; Springer, 2005; pp. 249--260. Available online: \url{https://doi.org/10.1007/11604686_22}.

\bibitem[Ishida~et~al.(2008)]{Ishida2008}
Ishida, Y.; Zhao, L.; Nagamochi, H.; Akutsu, T. Improved Algorithms for Enumerating Tree-Like Chemical Graphs with Given Path Frequency. \emph{Genome Inform.} \textbf{2008}, \emph{21}, 53--64. Available online: \url{https://doi.org/10.1142/9781848163324_0005}.

\bibitem[Shimizu~et~al.(2011)]{Shimizu2011}
Shimizu, M.; Nagamochi, H.; Akutsu, T. Enumerating Tree-Like Chemical Graphs with Given Upper and Lower Bounds on Path Frequencies. \emph{BMC Bioinform.} \textbf{2011}, \emph{12}, 1--9. Available online: \url{https://doi.org/10.1186/1471-2105-12-S14-S3}.

\bibitem[Suzuki~et~al.(2012)]{Suzuki2012}
Suzuki, M.; Nagamochi, H.; Akutsu, T. A 2-Phase Algorithm for Enumerating Tree-Like Chemical Graphs Satisfying Given Upper and Lower Bounds. \emph{IPSJ SIG Tech. Rep.} \textbf{2012}, \emph{28}, 1--8.

\bibitem[Akutsu and Fukagawa(2005)]{Akutsu2005}
Akutsu, T.; Fukagawa, D. Inferring a Graph from Path Frequency. In \emph{Proceedings of the 16th Annual Symposium on Combinatorial Pattern Matching (CPM)}, 2005; pp. 371--382. Available online: \url{https://doi.org/10.1007/11496656_32}.

\bibitem[Ru{\'e}~et~al.(2014)]{Rue2014}
Ru{\'e}, J.; Sau, I.; Thilikos, D. M. Dynamic Programming for Graphs on Surfaces. \emph{ACM Trans. Algorithms} \textbf{2014}, \emph{10}, 1--26. Available online: \url{https://doi.org/10.1145/2556952}.

\bibitem[Imada~et~al.(2011)]{Imada2011}
Imada, T.; Ota, S.; Nagamochi, H.; Akutsu, T. Efficient Enumeration of Stereoisomers of Tree-Structured Molecules Using Dynamic Programming. \emph{J. Math. Chem.} \textbf{2011}, \emph{49}, 910--970. Available online: \url{https://doi.org/10.1007/s10910-010-9789-9}.

\bibitem[Imada~et~al.(2010)]{Imada2010}
Imada, T.; Ota, S.; Nagamochi, H.; Akutsu, T. Efficient Enumeration of Stereoisomers of Outerplanar Chemical Graphs Using Dynamic Programming. \emph{J. Chem. Inf. Model.} \textbf{2010}, \emph{50}, 1667--1679. Available online: \url{https://doi.org/10.1021/ci200084b}.

\bibitem[Masui~et~al.(2009)]{Masui2009}
Masui, R.; Shurbevski, A.; Nagamochi, H. Enumeration of Unlabeled Trees by Dynamic Programming. \emph{Technical Report} \textbf{2009}.

\bibitem[He~et~al.(2016)]{He2016}
He, F.; Hanai, A.; Nagamochi, H.; Akutsu, T. Enumerating Naphthalene Isomers of Tree-Like Chemical Graphs. In \emph{Proceedings of the 7th International Conference on Bioinformatics Models, Methods and Algorithms (BIOSTEC 2016)}, Volume 3: BIOINFORMATICS, 2016; pp. 258--265. Available online: \url{https://doi.org/10.5220/0005783902580265}.

\bibitem[Azam et al.(2021)]{Azam2021}
Azam, N.A.; Zhu, J.; Sun, Y.; Shi, Y.; Shurbevski, A.; Zhao, L.; Nagamochi, H.; Akutsu, T. A novel method for inference of acyclic chemical compounds with bounded branch-height based on artificial neural networks and integer programming. {\em Algorithms Mol. Biol.} {\bf 2021}, {\em 16}, 18.

\bibitem[Azam~et~al.(2020c)]{Azam2020b}
Azam, N. A.; Zhu, J.; Ido, R.; Nagamochi, H.; Akutsu, T. Experimental Results of a Dynamic Programming Algorithm for Generating Chemical Isomers Based on Frequency Vectors. In \emph{Proceedings of the Fourth International Workshop on Enumeration Problems and Applications (WEPA)},\textbf{2020}, Paper ID 15. 

\bibitem[Akutsu and Nagamochi(2020)]{Akutsu2020}
Akutsu, T.; Nagamochi, H. A Novel Method for Inference of Chemical Compounds with Prescribed Topological Substructures Based on Integer Programming. \emph{arXiv} \textbf{2020}, arXiv:2010.09203. Available online: \url{https://arxiv.org/abs/2010.09203}.

\bibitem[Jordan(1869)]{Jordan1869}
Jordan, C. Sur les Assemblages de Lignes. \emph{J. Reine Angew. Math.} \textbf{1869}, \emph{70}, 185--190. Available online: \url{https://doi.org/10.1515/crll.1869.70.185}.

\bibitem[Azam~et~al.(2020a)]{Azam2020c}
Azam, N. A.; Shurbevski, A.; Nagamochi, H. Enumeration of Tree-Like Graphs with a Given Cycle Rank Using Dynamic Programming. \emph{Entropy} \textbf{2020}, \emph{22}, 1295. Available online: \url{https://doi.org/10.3390/e22111295}.

\bibitem[On the Enumeration of Polymer Topologies(2017)]{Haruna2017}
On the Enumeration of Polymer Topologies. \emph{Report} No. 2017-AL-162, March 2017.
\end{thebibliography}
\end{document}